\documentclass[12pt]{amsart}
\usepackage[utf8]{inputenc}

\usepackage[left=3.00cm, right=3.00cm, top=3.00cm, bottom=3.00cm]{geometry}

\usepackage{amsfonts,amssymb,amsmath,amsthm}
\usepackage{amsmath,amsfonts}
\usepackage{graphicx}
\usepackage{xspace}
\usepackage{xfrac}

\usepackage[authoryear]{natbib}

\newtheorem{corollary}{Corollary}
\newtheorem{example}{Example}
\newtheorem{remark}{Remark}
\newtheorem{definition}{Definition}
\newtheorem{theorem}{Theorem}
\newtheorem{proposition}{Proposition}
\newtheorem*{lemma}{Lemma}

\newtheorem*{acknowledgements}{Acknowledgements}

\newcommand{\bydef}{:=}
\newcommand{\N}{\mathbb{N}}
\newcommand{\R}{\mathbb{R}}
\newcommand{\simplex}{\Delta}

\newcommand{\A}{\ensuremath{\mathbb{A}}\xspace}
\newcommand{\B}{\ensuremath{\mathbb{B}}\xspace}
\newcommand{\X}{\ensuremath{\mathbb{X}}\xspace}

\newcommand{\be}{\mathbf{e}}
\newcommand{\ba}{\widetilde{\mathbf{a}}} 
\newcommand{\bb}{\widetilde{\mathbf{b}}} 

\newcommand{\bF}{\mathbf{F}}
\newcommand{\bN}{\mathbf{N}}
\newcommand{\bM}{\mathbf{M}}
\newcommand{\bL}{\mathbf{L}}

\newcommand{\bA}{\mathbf{A}}
\newcommand{\bB}{\mathbf{B}}

\newcommand{\bP}{\mathbf{P}}

\newcommand{\bT}{\mathbf{T}}
\newcommand{\bX}{\mathbf{X}}
\newcommand{\bG}{\mathbf{G}}
\newcommand{\bzero}{\mathbf{0}}
\newcommand{\bu}{\mathbf{u}}
\newcommand{\bv}{\mathbf{v}}
\newcommand{\bp}{\mathbf{p}}
\newcommand{\bq}{\mathbf{q}}
\newcommand{\bI}{\mathbf{I}}

\newcommand{\bone}{\mathbf{1}}
\newcommand{\bx}{\mathbf{x}}

\newcommand{\E}{\mathbb{E}}

\newcommand{\SO}{\ensuremath{\mathsf{StO}}}
\newcommand{\SSO}{\ensuremath{\mathsf{St^2O}}}
\newcommand{\BSO}{\ensuremath{\mathsf{BStO}}}

\newcommand{\RCP}{\textsf{RCP}}

\newcommand{\K}{\ensuremath{\mathcal{K}}}
\newcommand{\G}{\ensuremath{\mathcal{G}}}

\newcommand{\fd}{\mathfrak{d}}

\title{
	On the stochastic evolution of finite populations
}

\author{Fabio A. C. C. Chalub$^*$} 

\address{
	$^*$Departamento de Matemática and Centro de Matemática e Aplicações, Faculdade de Ciências e Tecnologia, Universidade Nova de Lisboa, Quinta da Torre, 2829-516, Caparica, Portugal.}

\email{chalub@fct.unl.pt}

\author{Max O. Souza$^\dagger$}
\address{
	$^\dagger$Departamento de Matem\'atica Aplicada, Universidade Federal Fluminense, R. M\'ario Santos Braga, s/n, 22240-920, Niter\'oi, RJ, Brasil.
}

\begin{document}

\maketitle

\begin{abstract}
This work is a systematic study of  discrete Markov chains that are used to describe the evolution of a two-types population. Motivated by results valid for the well-known  Moran (M) and  Wright-Fisher (WF) processes, we define a  general class of Markov chains models which we term the Kimura class. It comprises the majority of the models used in population genetics, and we  show that many  well-known results valid for M and WF processes  are still valid in this class. In all Kimura processes, a mutant gene will  either fixate or become extinct, and we present a necessary and sufficient condition for such processes to have the probability of fixation strictly increasing in the initial frequency of mutants. This condition implies that there are WF processes with decreasing fixation probability --- in contradistinction to M processes which always have strictly increasing fixation probability. As a by-product, we show that an increasing  fixation probability  defines uniquely an   M  or WF  process which realises it, and that  any fixation probability with no state having trivial fixation can be realised by at least some WF process. These results are extended to a subclass of processes that are suitable for describing time-inhomogeneous dynamics.
We also discuss the traditional identification of frequency dependent fitnesses and pay-offs, extensively used in evolutionary game theory, the role of weak selection when the population is finite, and the relations between jumps in evolutionary processes and frequency dependent fitnesses.


\end{abstract}

\section{Introduction}

The evolution of finite populations  is inherently prone to  stochastic effects, and these are enhanced if the population is small. Thus, the correct modelling of these effects is a key step in understanding the dynamics of such populations. 
 The Wright-Fisher process, one of the most prevalent models in mathematical population genetics,  was a watershed point, and has set up much of the current paradigm in modelling finite populations \citep{Fisher1,Wright1}. This process is a Markov chain with a finite number of states that are  distributed multinomially, and hence not easily amenable to analysis except in some special cases~\citep{Fisher1,Wright_1931}. As a simpler process that was assumed to capture the essential aspects of the Wright-Fisher dynamics,  Moran conceived a Birth-Death process for genetics, that is contemporaneously known as the Moran model~\citep{Moran}.  Later on, one can find a plethora of processes designed to model the evolution of genotype frequencies---see for instance~\citep{KimuraCrow,Ewens_2004,EthierKurtz,Charlesworth,Hartle_Clark} and references therein for many examples with different levels of rigour and generality.

We consider only populations with two types of individuals, and within this framework our archetypical examples will be the Moran and Wright-Fisher processes.  The former, as observed above, is a birth-death process, and hence one individual is replaced at each time step. On the other hand, the entire population is replaced at once in the latter. Both processes are  Markov chains with two absorbing states, which are the only stationary states.  They  also share the same diffusion approximation---up to rescaling---for large populations and weak-selection, which suggests that, at least in this regime, they are two sides of the same coin~\citep{McCandlishEpsteinPlotkin}. However, to the  best of our knowledge,  it has not been examined so far if these similarities extend outside this regime.

From a broader perspective, both Moran and Wright-Fisher processes --- when considered in a population comprising two types, which will denote by \A and \B --- belong to  a class of Markov chains that are characterised by two parameters: the population size $N$, and a vector of $N+1$ \textit{type selection probabilities},  where each entry  indicates how likely   type \A is to be chosen for reproduction depending on  its prevalence in the population --- and hence it accounts for the effects of natural selection in the model. Since mutations effects are not considered, the homogeneous states, i.e., states with full absence or full  prevalence of a given type, are absorbing states. In the former case, this type has become extinct, while in the latter it has fixated.  In addition,  every state of the population is accessible from any non-homogeneous state in a finite number of steps.  This class of models appears in so many instances, that is natural to name it: we will term it the \emph{Kimura} class.

The axiomatisation of this  class will be our starting point, and this will lead to a number of questions that seem to  be unnoticed in the literature:

\begin{description}
 
\item[\textsl{\textbf{Fitness and type selection probabilities}}:] 
One of the most typical modelling approach is to consider that the probability of a given type will be selected for reproduction is proportional to chosen function, that depends on the state of the population. These functions, one for each type, are a modelling proxy for their reproductive success, and they are usually identified with the fitness of the corresponding type. Here, we will be somewhat more careful, and will term them reproductive fitness.  In this setting, a natural question to consider is whether this notion  of fitness is consistent with  the classical one of  Darwinian fitness --- the ratio between the prevalence of a given type in two successive generations --- which may be computed directly from the statistical properties of the corresponding process. We will show that such  identification  is always consistent with the Wright-Fisher process, but not  with any birth-death process; in particular, it is not consistent with Moran processes.

\item[\textsl{\textbf{Qualitative properties of fixation}}:]	
For models without mutation, a  key quantity in understanding their evolutionary behaviour  is how likely a given type will fixate as a function of its current frequency in the population. For neutral processes,  the fixation probability of any type is given by the current  frequency of this type. In particular, the larger the fraction of individuals of a certain type, the larger is the fixation probability of this particular type. Such a monotonic behaviour of the fixation  probability is usually taken for granted (e.g.~\citep{der2011generalized,Nassar_Cook}) also outside the neutral regime.  This is certainly correct in the diffusive limit~\citep{ChalubSouza09b}, but otherwise this issue seems to be largely overlooked.  We point out that the explicit expression for the fixation probability for birth-death processes implies that it is strictly increasing in the initial frequency --- a phenomenon that we term \emph{regular evolution}. We then show that the Moran process is a universal regular process: given any increasing fixation vector there exists a unique choice of type selection probabilities for which the corresponding Moran process has the given fixation. We also show that the Wright-Fisher (WF) is a universal process as far as fixation of probability is concerned. Namely, given an admissible fixation vector $\bF$, i.e. a fixation for which the only entry with a zero is the  first, and  the only  entry with a one is the last, then there exists a WF process that  has $\bF$ as a fixation vector. Furthermore, if the fixation is increasing,  this process is unique. This shows, at least in some cases, that the fixation probability can completely characterise the process, which extends the characterisation of neutrality through fixation probabilities as done in, e.g., \citep{Kimura_1962,Hartle_Clark,Fontdevila}.

\item[\textsl{\textbf{Time-inhomogeneous evolution}}:]
All processes described thus far are homogeneous in time. Nevertheless, real environments are not static, and the understanding of  evolutionary dynamics in this changing scenario provides a new set of challenges, as evidenced by the recent results of~\citet{Ashcroft:etal:2014,Melbinger:Vergassola:2015,Uecker,Feldman,Cvijovic}. While matrices in the Kimura class have the appropriate stochastic properties to describe  a single step of evolutionary processes, we  show that it lacks the appropriate structure to describe time-inhomogeneous processes, as it is not closed by multiplications. Therefore, it needs to be restricted, and at this point we will introduce the Gillespie class, $\G$, which is still large enough to include all previous examples. This new class is a convex semigroup --- i.e. it is closed by convex combinations and products---and it corresponds to processes whose transient dynamics is given by a totally indecomposable matrix rather than an irreducible one.

\end{description}

 In Section~\ref{sec:MandWF} we introduce the  Kimura class, \K, and in particular the concept of type selection probability. We also point out that both Wright-Fisher and Moran processes belong to this general class. We then introduce the parametrisation of type selection probabilities by reproductive fitness, and show that these are consistent with  Darwinian fitness  if, and only if, the type selection probabilities are the expected frequency in the next generation. It turns out that this condition is satisfied by the Wright-Fisher process but not by any birth-death process. We finish this section showing a few general results on the fixation probability that while seem to be known are not conveniently available. 

In Section~\ref{sec:regular}, we introduce the concept  of regular evolution as discussed above and three classes of progressively more general ordered matrices: strictly stochastically ordered, banded stochastically ordered and stochastically ordered. We characterised the regular process as all processes that are eventually stochastically ordered. Indeed, we show that regularity is equivalent to a certain transformed matrix of the process having the so-called Perron-Frobenius property.

In Section~\ref{sec:fix_reg_classical}, we study the Moran process and, by means of the expression for the fixation probability of a birth-death process, we conclude that it is regular. We also show that, given a increasing fixation vector, there exists a unique choice of type selection probabilities for which the corresponding Moran process has this fixation vector. In the sequel, we proceed to the study of the Wright-Fisher process, and show that being regular is equivalent to have the type selection vector increasing and also equivalent to being strictly stochastically ordered. As a by product, we show that for any  admissible fixation, there exists a choice of type selection probabilities --- not necessarily unique --- such that this fixation is realised by the corresponding Wright-Fisher process. Uniqueness of such a choice holds, if the fixation vector is increasing. The possibility of non-regular evolution for the Wright-Fisher process is further investigated within the realm of evolutionary game theory. In this vein, we prove that type selection probabilities given by two-player game theory are always increasing, with the usual identification of fitnesses and pay-offs; hence the corresponding process is always regular. On the other hand, as soon as we move to three-player games we can find examples of non-increasing type-selection probabilities, and hence non-regular Wright-Fisher dynamics. Finally, we give a complete characterisation of regularity for matrices in the Kimura class.  We finish with the study of what happens with type selection probabilities, if we assume that fixation is given by a smooth function, and the population is large, but still finite.

In Section~\ref{sec:inhomogeneous} we deal with time-inhomogeneous processes. We introduce  the Gillespie class, $\G$, which is contained in the Kimura class, and show that it is closed under products and convex combinations, and thus it is a convex semigroup. We then study how time-inhomogeneity  might affect regularity. 
We show that  a time-inhomogeneous Wright-Fisher process that is locally regular --- i.e. a process such that every transition matrix between two consecutive time steps is given by a regular Wright-Fisher process --- is itself regular.
On the other hand, no such a result holds for the Moran process (or birth-death processes in general). Indeed, we provide an example of two Moran matrices whose product is non-regular.  This yields a deterministic version  of Parrondo's paradox.  

We finish with a discussion of the results in Section~\ref{sec:discussion}.

\begin{table}
\centering
\begin{tabular}{p{5cm}|p{5cm}|c}
\textbf{Name} & \textbf{Symbol} & \textbf{Definition}\\
\hline
Kimura matrix & $\bM\in\K $& Def.~\ref{def:Kimura}\\
Regular process & $\bM$ & Def.~\ref{def:regular}\\
Gillespie matrix & $\bM\in\G$ & Def.~\ref{def:Gillespie}\\
Core matrix & $\widetilde{\bM}$ & Def.~\ref{def:Kimura}\\
Associated matrix& $\bL$& Subsec.~\ref{ssec:regular_equiv}\\
Fixation probability & $\bF$ & Eq.~(\ref{eq:Fisfix})\\
Type Selection Probability (TSP) & $\bp$ (full population), $\bq$ (restricted population) & Subsec~
\ref{ssec:natural_class}\\
(Strictly, banded) Stochastic ordered matrices & (\SSO, \BSO) \SO & Def.~\ref{def:som}\\
Darwinian fitness & $\Psi$ (absolute) and $\Phi$ (relative) & Def.~\ref{deff:dawin_fit}\\
Reproductive fitness & $\psi$ & Eq.~(\ref{eq:p_and_fit_rep})\\
\end{tabular}
\caption{Table with some important notation introduced in this work.}
\label{table:notation}
\end{table}

\section{The basic set-up}
\label{sec:MandWF}

We use italics to denote real numbers, while  boldface denote either vectors in $\R^{N+1}$ or matrices in $\R^{(N+1)\times (N+1)}$. Furthermore, vectors in $\R^{N-1}$ and matrices in $\R^{(N-1)\times(N-1)}$ are denoted by bold-tilded symbols.  There are two exceptions to these conventions: (i) the associated matrix $\bL$ is defined in $\R^{N\times N}$; (ii) all null vectors are written without tildes. For the convenience of the reader, Table~\ref{table:notation} summarises the notation used throughout the text.

We use the probabilist convention that vectors are row matrices. Hence matrices act on vectors through multiplication from the right, and on  transposed vectors through multiplication from the left.

Finally, recall that a non-negative matrix $\mathbf{A}\in\R^{N\times N}$ is \emph{stochastic} if $\sum_{j=1}^NA_{ij}=1$ for all $i$ and that it is \emph{sub-stochastic}, if $\sum_{j=1}^NA_{ij}\le 1$ for all $i$, with strict inequality holding for at least one value of $i$.

\subsection{A natural class of Markov chains}
\label{ssec:natural_class}

 Consider a population of two types denoted by \A and \B, respectively, with fixed size $N\in\N$, which evolves in the absence of mutation. We will assume that the population dynamics is described by a discrete time Markov chain, with time-homogeneous transition probabilities, and where the chain is in state $j$, if there are $j$ individuals of type \A in the population.

As is well known \citep{Karlin_Taylor_first,Karlin_Taylor_intro}, such a chain can be completely described by specifying the transition probability from state $i$ to $j$, which we will denote by $M_{ij}$, $i,j=0,\ldots,N$. We will further assume that $M_{ij}=M(i,j,\bp,N)$, 
where $\bp\in\R^{N+1}$ is a vector of probabilities such that $p_j\in[0,1]$ describes the probability of a type \A individual being  selected  for reproduction, with the chain in state  $j$. We will term $\bp$ the vector of \textit{type selection probabilities} (TSP).

 Since there is no mutation, we have  that
 \[
 M_{0i}=\left\{\begin{array}{ll}
 1,& i=0\\ 0,&i=1,\ldots,N
 \end{array}\right.
 \quad\text{and}\quad
 M_{Ni}=\left\{\begin{array}{ll}
 0,&i=0,\ldots,N-1\\
 1,& i=N
 \end{array}\right..
 \]
 In agreement with this interpretation of $\bp$, we will always have $p_0=0$, $p_N=1$, and $0<p_i<1$ for $i=1,\dots,N-1$, unless stated otherwise. Furthermore, we assume that all sates are accessible to the dynamics, from any non absorbing initial condition, in a sufficiently large number of steps. 
 
All these assumptions combined lead to the following definition:

\begin{definition}[The Kimura class of matrices]
	
	\label{def:Kimura}
Let $\bM$ be a $(N+1)\times(N+1)$ stochastic matrix. We say that  $\bM$ is  \emph{Kimura}, if
 \begin{equation}\label{eq:Mdeff}
 \bM=
 \begin{pmatrix}
 1&\bzero&0\\
 \ba^\dagger&\widetilde{\bM}&\bb^\dagger\\
 0&\bzero&1
 \end{pmatrix},
 \end{equation}
where $\widetilde{\bM}$ is a $(N-1)\times(N-1)$ sub-stochastic irreducible matrix, $\bzero$ is the zero vector in $\R^{N-1}$, and with $\ba$ and $\bb$  non-zero, non-negative  vectors in $\R^{N-1}$. We  will also say that  $\widetilde{\bM}$ is the \emph{core} matrix of $\mathbf{M}$, and that $\ba$ is the vector of one-step extinction rate, and $\bb$ is the vector one-step fixation rate.	The class of Kimura matrices is denoted by \K.

\end{definition}
Notice that  since the sum of an irreducible matrix with a non-negative matrix is itself irreducible, we immediately have that the Kimura class is convex.

\begin{remark}
We termed the matrices in definition~\ref{def:Kimura} after the celebrated sentence from the geneticist Motoo Kimura that states that ``A  mutant gene which appeared in a finite population will eventually either be lost from the population or fixed (established) in it''~\citep{Kimura_Ohta_1969}. The Kimura theory, however, is also know from the importance given to the neutral and quasi-neutral evolution (see~\citep{Kimura1983}). Here, we depart as much as possible from the neutral theory.
\end{remark}

This class can be  described as time homogeneous Markov chains with two absorbing states, and in which all states are reachable from all non-absorbing states in a fixed and finite number of steps. It is implicitly the standard class used in many applications (see, e.g., the discussion on the Wright-Fisher model without mutations at~\citet{Charlesworth} or, more generally, Markov models with absorbing states at~\citet{Karlin_Taylor_intro}). When considering time-inhomogeneous processes, this class will have to be restricted; see section~\ref{sec:inhomogeneous}.

Finally, we note that in this work, all $(N+1)\times(N+1)$ transition matrices, except if otherwise stated, will be Kimura.

\subsection{Two archetypical evolutionary processes}
\label{ssec:archetypical}

Two important evolutionary processes that fit the framework given above are the generalised {Moran and generalised Wright-Fisher processes, that are defined as follows:

\begin{description}
\item[Moran process (M)]
The Moran process is  a death-birth process \citep{Moran} --- hence with overlapping generations. Given a population with $N$ individuals at state $i$,  one individual is selected to die with probability $1/N$, and he/she is replaced by an individual of type \A with a  probability $p_i$ and a type \B individual with probability $1-p_i$. Therefore, the transition matrix is given by
\begin{equation*}
M_{ij}=\left\{
\begin{array}{ll}
\frac{i}{N}(1-p_i)\ ,\quad&i=j+1\ ,\\
\frac{i}{N}p_i+\frac{N-i}{N}(1-p_i)\ ,&i=j\ ,\\
\frac{N-i}{N}p_i\ ,&i=j-1\ ,\\
0\ ,&|i-j|>1\ .
\end{array}
\right.
\end{equation*}

\item[Wright-Fisher process (WF)]

The Wright-Fisher model is an evolutionary  process with non-overlapping generations \citep{Fisher1,Fisher2,Wright1,Wright2}. Given  a population of constant size $N$ at state $i$, it is replaced by a new population, where the probability that the new generation is at state $j$ is given by  
\[
M_{ij}=\binom{N}{j}p_i^j(1-p_i)^{N-j}\ .
\]
\end{description}

Let $X_k$  be the number of type \A individuals in the population, under the corresponding evolutionary process. Then, the expected number of type \A individuals, given the state $i$, in the next step is given by
$\E\left[X_{k+1}|X_k=i\right]= \sum_{j=0}^NjM_{ij}$.

\begin{definition}
We say that an evolutionary process  is \emph{neutral} if the corresponding population process, $X_k$, satisfies $\E\left[X_{k+1}|X_k=i\right]=i$, for all $i$.
Indeed, neutrality as defined by means of an  exchangeable model  is equivalent, for a two types populations, to the property that the population process of a given type is a martingale---cf. \citep{Cannings:1974,Cannings:1975}.
\end{definition}

Notice that some references define neutrality from properties of the type selection probability vector $\bp$~\citep{Ewens_2004,Burger_2000,der2011generalized} or from the fixation probability always being equal to the initial frequency~\citep{Kimura_1962,Hartle_Clark,Fontdevila}. Here we opt for the definition from the Darwinian fitness, i.e, that their conditional expected frequency is the same as the observed one~\citep{Nowak:06,gillespie1991causes}. For a population with two types, these definitions are equivalent, as we will see below.

For non-neutral processes, the expected number of type \A individuals after one time-step will be different from its current frequency. Indeed, for the Moran process we have that 
\begin{align}
\nonumber
		\E\left[X^{M}_{k+1}|X^{\mathrm{M}}_k=i\right]&=(i-1)M_{i,i-1}+iM_{ii}+(i+1)M_{i,i+1}=\\
		\nonumber
 		&=\frac{i(i-1)+(N-i)i}{N}+\frac{i^2+(i+1)(N-i)-i(N-i)-(i-1)i}{N}p_i\\	
		&=i+p_i-\frac{i}{N}\ .\label{eqn:exp_m}
		\end{align}
In the case of the WF process, the calculation is immediate from known properties of the binomial distribution: 
\begin{equation}\label{eqn:exp_wf}
 \E\left[X^{WF}_{k+1}|X^{\mathrm{WF}}_k=i\right]=Np_{i}.
\end{equation}

Notice that Equation~\eqref{eqn:exp_wf} further implies, for the WF process, that $p_{i}$ is the expected frequency of type \A in the population in the next generation, given that in the present generation such frequency is $\sfrac{i}{N}$.

\subsection{Type selection probabilities and Darwinian fitness}	

As observed by \citet{Orr:2009}, biologists broadly agree on the essence of the idea of fitness, although they give many different definitions. Fundamentally, fitness measures the ability of individuals to survive and reproduce in some environment. This can be measured through the expected number of type \A individuals in the next generation, from a given state--- this is the \emph{Darwinian fitness} \citep{Hartle_Clark}. On the other hand, in many different modelling approaches --- as in evolutionary game theory, where fitness is identified with the game pay-off --- fitness is used as a proxy to the probability to select a given type, when determining the next reproduction event; see, for example~\citep{Imhof:Nowak:2006}. Here, we call this function \emph{reproductive fitness}.

  Let us denote the reproductive fitnesses by $\varphi^{(\A,\B)}:\{0,1,\dots,N\}\to\R_+$. Then, these  fitness functions are typically related to the type selection probabilities by
  \begin{equation}\label{eq:p_and_fit_rep}
  \bp_{i}=\frac{i\varphi^{(\A)}(i)}{i\varphi^{(\A)}(i)+(N-i)\varphi^{(\B)}(i)}.
  \end{equation}
This assumes that the population is well-mixed, and it is in agreement with the  intuitive idea that   $\sfrac{\varphi^{(\A)}}{\varphi^{(\B)}}$  indicates the probability to be selected for reproduction.

Equation~(\ref{eq:p_and_fit_rep}) is widely used in many works dealing either with the Wright-Fisher process (see, e.g.,~\citep{Imhof:Nowak:2006,AntalScheuring,TraulsenPachecoImhof,ChalubSouza_2016}) or specifically with the Moran process \citep{Nowak:etal:2004,Nowak:06,ChalubSouza09b,TraulsenPachecoImhof,Fudenbergetal_2006}.  Nevertheless, we will now see that the identification between Darwinian and reproductive fitnesses is consistent  only with the Wright-Fisher process for all possible cases.  	We follow~\cite{Hartle_Clark} in the definition below. See also~\citep{KimuraCrow,MaynardSmith}

\begin{definition}\label{deff:dawin_fit}
Let $X_k^*$ be a stochastic  population process, with finite population $N$. Let us define the (discrete time) \emph{Darwinian fitness} fitness as
\[
 \Psi^\A(i)\bydef\frac{\E\left[X_{k+1}|X_k=i\right]}{i}\ ,\qquad 
 \Psi^\B(i)\bydef\frac{\left(N-\E\left[X_{k+1}|X_k=i\right]\right)}{(N-i)}
\]
The relative Darwinian fitness is 
\[
 \Phi(i)\bydef\frac{\Psi^\A(i)}{\Psi^\B(i)} .
\]
\end{definition}

\begin{lemma}
	An evolutionary stochastic process with finite population is consistent both  with reproduction fitness  and with  Darwinian fitness if, and only if, the corresponding population process satisfies:
	\[
	\E\left[X_{k+1}|X_k=i\right]=Np_{i}.
	\]
\end{lemma}	
\begin{proof}
	From the Darwinian fitness definition, we have
	\[
	\Phi(i)=\frac{N-i}{i}\frac{\E\left[X_{k+1}|X_k=i\right]}{N-\E\left[X_{k+1}|X_k=i\right]}
	\]
	while from the reproductive fitness parametrisation, we have that
	\[
	\Phi(i)=\frac{N-i}{i}\frac{p_{i}}{1-p_{i}},
	\]
	and the result follows.
\end{proof}

\begin{remark}
	The above result together with formulas~(\ref{eqn:exp_m}) and~(\ref{eqn:exp_wf}) implies that the Moran process is consistent  with both reproductive fitness  and with  Darwinian fitness if, and only if, it is neutral. More precisely, when using the Moran process to describe the evolution of a population,  we may not identify  the Darwinian fitness and the reproductive fitness. In this vein, there are two possibilities: alternatively	
		\begin{enumerate}
			\item One can use the reproductive fitness for modelling,  but should not expect that the type selection probabilities determine the expected value of type \A in the next generation.
			\item One can  calibrate a Moran model by the Darwinian fitness. However, in this case the type selection probabilities are given by 
					\begin{equation}\label{eq:Mp}
					p_{i}=\frac{i}{N}+\frac{i(N-i)(\Phi(i)-1)}{N-i+i\Phi(i)}.
					\end{equation}
		\end{enumerate}
\end{remark}

 Local in time maximization of reproductive fitness implies maximization of the fixation probability if we are in the weak selection regime and if population is large. See, e.g., \citep{ChalubSouza09b} for the Moran process and~\citep{ChalubSouza14a} for the Wright-Fisher dynamics, where it is shown that (using game-theory vocabulary) Nash-equilibrium strategy maximizes fixation probability, using the traditional identification between pay-offs and reproductive fitness. If the population is small, however, the situation can be strikingly different, as shown by the example below. 
 \begin{example}
 Consider the Public Good Game: $N$ players have to contribute or not one monetary unit for a common pool, which is afterwards multiplied by $r>0$ and split equally among all players, irrespectively of their personal strategy, cf.~\citep{Archetti_JTB2012}. For any value of $r$, donors have a smaller pay-off of exactly one unity than non-donors, and therefore any evolutionary dynamics that equals reproductive fitness with pay-off will favour the fixation of non-donors. However the behaviour of \emph{rational} (i.e., pay-off maximizer) players will depend on the value of $r$. If $r<N$ (the most traditional setting) rationally players will not contribute, while for $r>N$ rational players will contribute as the net return is $\sfrac{r}{N}>1$ per monetary unity. Therefore, evolutionary dynamics will result in a population of non-rational individuals, i.e., the final outcome of the evolutionary dynamics will be a non-Nash equilibrium. Note that in this simple setting, the meaning of ``small'' and ``large'' population is clear, as the multiplicative parameter $r$ introduces a natural scale in the model. This model is an example of Hamiltonian spite~\citep{Hamilton_1970} in which individuals maximize their fixation probability  minimizing their pay-off.
\end{example}

\subsection{Fixation probabilities}

	We recall that the vector of fixation of probabilities is defined by
		\begin{equation}\label{eq:Fisfix}
		F_i=\lim_{\kappa\to\infty}\langle \be_i\bM^\kappa,\be_N\rangle.
		\end{equation}
		The entry $F_i$  is the fixation probability of type \A from a population with $i$ individuals of type \A.
		
		We also recall---cf. \citep{Karlin_Taylor_first,Karlin_Taylor_intro}---that 
			\begin{equation}\label{eq:limit}
			\bM^\infty\bydef		\lim_{k\to\infty}\bM^k=(\bone-\bF)\otimes\be_0+\bF\otimes\be_N.
			\end{equation}

Another characterisation of the fixation vector is as a left eigenvector of the transition matrix associated to the eigenvalue one, with the first entry being zero,  and normalised such that the last entry is one \citep{Karlin_Taylor_intro,Karlin_Taylor_first}. It can be also characterised in a more algebraically fashion---cf. \citep{GS:1997,Karlin_Taylor_first}---as follows
\begin{proposition}
	\label{prop:fixationvector}
	Let $\bM\in\K$. Then, there exists a unique vector $\widetilde{\bF}\in\R^{N-1}$, with $0<\widetilde{F}_i<1$,  such that $\bF^=\begin{pmatrix}0&\widetilde{\bF}&1\end{pmatrix}$, with $\bM\bF^\dagger=\bF^\dagger$ and 
	\[
	\widetilde{\bF}^{\dagger}=\left(\bI-\widetilde{\bM}\right)^{-1}\bb^{\dagger}.
	\]

\end{proposition}

\begin{proof}
	A straightforward computation shows that $\bM\bF^{\dagger}=\bF^{\dagger}$ is equivalent to 
	\[
\begin{pmatrix}
 0\\
 \widetilde{\bM}\widetilde{\bF}^{\dagger}+\bb^{\dagger}\\
 1
\end{pmatrix}
=
\begin{pmatrix}
 0\\
 \widetilde{\bF}^{\dagger}\\
 1
\end{pmatrix}.
	\]
	This will be satisfied if, and only if, we have
	\begin{equation}\label{eq:Ffromb}
	\left(\bI-\widetilde{\bM}\right)^{-1}\bF^{\dagger}=\bb^{\dagger}.
	\end{equation}
	Since $\widetilde{\bM}$ is sub-stochastic and irreducible, $\bI-\widetilde{\bM}$ is invertible and $\left(\bI-\widetilde{\bM}\right)^{-1}$ is positive---cf. \citep{BP79}. Hence $\widetilde{\bF}$ is uniquely defined, and positive. A similar calculation with $\bone-\bF$ instead of $\bF$ yields
	\begin{equation}\label{eq:Ffroma}
\left(\bI-\widetilde{\bM}\right)\left(\bone-\widetilde{\bF}\right)^{\dagger}=\ba^\dagger.
	\end{equation}
	Hence, using the same argument as above, we conclude that  $\bone-\widetilde{\bF}$ is positive, and hence that  we have $\widetilde{F}_i<1$.

\end{proof}

The vector $\bF$ and, by extension, the vector $\widetilde{\bF}$ are called the \emph{fixation vectors} associated to the process $\bM$.

\begin{definition}[Admissible fixation vector]
A fixation vector $\bF$ satisfying $0<F_i<1$, $i=1,\dots,N-1$, is termed admissible. Proposition~\ref{prop:fixationvector} then states that any fixation vector from a process whose transition matrix is Kimura is admissible. 
\end{definition}

\begin{remark}\label{rmk:L1}
It is possible  for a process not in the Kimura class to fixate.  Indeed, let $\K_0$ the set of matrices with the form given by Eq.~(\ref{eq:Mdeff}) with $\rho(\widetilde{\bM})<1$, where $\rho(\bA)$ denotes the spectral radius of $\bA$. Then it is easy to see that a process without mutation fixates if, and only if, its transition matrix belongs to $\K_0$, although  the corresponding fixation may not be admissible. On the other hand, if we write $\K_1\subset\K_0$ with $\ba,\bb>\bzero$, then every matrix in $\K_1$ has an admissible fixation.  It can be also shown that $\K_1$ is closed through convex combinations and multiplications.  Although $\K_0$ and $\K_1$ are not our primary interest in this work, most of the results presented here could be easily extended to $\K_1$, and some even to $\K_0$. In subsection~\ref{ssec:further}, we study examples that belong to these classes. 
\end{remark}

\section{Regular Fixation}
\label{sec:regular}

\subsection{Preliminary definitions and results}

We say that $\bu\in\R^N$ is \emph{increasing} (non decreasing) if for all $i>j$, we have  $u_i>u_j$ ($u_i\ge u_j$).

\begin{definition}[Regular and weakly-regular processes]\label{def:regular}
	We say that an  evolution process such that the transition matrix belongs to the Kimura class is  \textbf{regular} (\textbf{weakly regular}), if the associated  fixation vector is  increasing (non-decreasing, respect.). By extension, we shall say that an  increasing (a non-decreasing, respect.) fixation vector is a regular (weakly-regular, respect.) fixation. Notice that any regular fixation is necessarily admissible. 
\end{definition}

\begin{remark}
	A regular evolutionary process conforms to the intuitive idea  that the probability of fixation of a given type  increases when  the number of individuals of this type increases---cf. \citep{der2011generalized,McCandlishEpsteinPlotkin,tan2012monotonicity}. 
\end{remark}

We begin by giving  a sufficient condition for a process to be regular. In order to do so,  we will need the concept of stochastic ordering of probability vectors:

\begin{definition}[Vector stochastic ordering]
	\label{def:sov}
	We say that two vectors $\bu,\bv\in\simplex^{N}\bydef\{\bx\in\R^{N+1}|x_i\ge 0,\sum_ix_i=1\}$ are \emph{stochastically ordered}, $\bu\succ\bv$, if for all $n=1,\dots,N$, we have that $\sum_{i=n}^{N}u_i\geq \sum_{i=n}^{N}v_i$. If all inequalities are strict, then we say $\bu\succ\!\!\succ\bv$.
\end{definition}
Matrices whose rows are stochastically ordered will also be relevant:
\begin{definition}[Ordered matrices]
 \label{def:som}
 Consider a $N\times N$ matrix $\bA$. We say that $\bA$ is \emph{stochastically ordered} (SO, $\bA\in\SO_N$) if all row vectors are stochastically ordered, i.e., if for all $i>j$, we have that $\bA_{i,\cdot}\succ\bA_{j,\cdot}$. 
	We say that $\bA$ is \emph{strictly stochastically ordered} (SSO, $\bA\in\SSO_N$) if for all $i>j$, we have that $\bA_{i,\cdot}\succ\!\!\succ\bA_{j,\cdot}$. 
	Finally, we say that  $\bA$ is a \emph{banded stochastically ordered matrix} (BSO, $\bA\in\BSO_N$), if it is stochastically ordered, and if $\sum_{j=n}^{N}A_{i+1,j}> \sum_{j=n}^{N}A_{i,j}$  for $n\in\{i+1-k_1,\dots,i+1+k_2\}\cap\{1,\dots,N\}$ for $k_1,k_2,\geq1$. We say that the matrix $\bA$ has \emph{upper (lower) band} of size $k_1$ ($k_2$, respec.). If $k_1=k_2$, we say that $\bA$ has a band of size $k_1$. 
\end{definition}

\begin{definition}[Eventually ordered matrices]
We say that a $N\times N$ matrix $\bA$ is \emph{eventually strictly stochastically ordered} (\emph{stochastically ordered}) if there exists $k_0\in\N$ such $A^k$ is strictly stochastically ordered (stochastically ordered, respect.) for $k\geq k_0$.
\end{definition}

	\begin{remark}
		\label{rmk:kkt}
  Following \citet{KK77}, we let  $\bT$ be a $(N+1)\times(N+1)$ matrix such that $T_{ij}=1$ for $i\ge j$ and $0$ otherwise. Therefore $T^{-1}_{i,i-1}=-1$, for $i=1,\dots,N$, $T_{ii}^{-1}=1$ for $i=0,\dots,N$ and $0$ otherwise. A vector $\bv$ is increasing if, and only if, $\bT^{-1}\bv$ has positive entries, except possibly for the first. Also, a matrix $\bM$ is stochastically ordered if, and only if, $\bT^{-1}\bM\bT$ is positive except possibly for the Oth column and row
  \end{remark}

We begin by identifying the relevant algebraic structure of these classes of matrices with a slight extension of a result already present in \cite{KK77}:

\begin{lemma}\label{lem:stocordered}
 Let  $\mathcal{M}_{N}$ be one of the sets $\SO_N$, $\SSO_N$ or $\BSO_N$. Then if $\bA,\bB\in\mathcal{M}_{N}$, then $\bA\bB\in\mathcal{M}_{N}$ and $\mu\bA+(1-\mu)\bB\in\mathcal{M}_{N}$ for all $\mu\in[0,1]$. Furthermore,
  $\SSO_N\subset\BSO_N\subset\SO_N$.
\end{lemma}

\begin{proof}
   Let the matrix $\bT$ be as defined in Remark~\ref{rmk:kkt}. It is clear that $\bT^{-1}\bA\bT$ is positive if and only if $\bA\in\SSO_N$, $\bT^{-1}\bA\bT$ is non-negative with the three principal diagonals positive if and only if $\bA\in\BSO_N$  and $\bT^{-1}\bA\bT$ is non-negative if and only if $\bA\in\SO_N$. The result follows noting that $\bT^{-1}\bX_1\bX_2\bT=\bT^{-1}\bX_1\bT\bT^{-1}\bX_2\bT$ and $\bT^{-1}(\mu\bX_1+(1-\mu)\bX_2)\bT=\mu\bT^{-1}\bX_1\bT+(1-\mu)\bT^{-1}\bX_2\bT$. 
\end{proof}

\begin{remark}
	\label{rmk:csg}
	A set  of matrices satisfying the conclusions of Lemma~\ref{lem:stocordered} is known as a convex semigroup, i.e., it is a convex set in a vector space, with an associative multiplication that is  compatible with the convex structure~\citep{Ressel1987,Berg}. Notice that it is  easy to give examples showing that $\K$ is not a convex semigroup, since the product of irreducible matrices is not necessarily irreducible---cf. \cite{BP79}. This implies that $\K$ is not a suitable class for time-inhomogeneous process, and an appropriate restriction will be discussed in Section~\ref{sec:inhomogeneous}.
\end{remark}

The terminology in the definitions above are not standard. Here we follow and generalise the definition at \citep[Section 3.5]{Nasell2011}; in~\citep{KK77} stochastically ordered matrices are called \emph{monotone} and if $\bu\succ\bv$ it is said that $\bu$  \emph{dominates} $\bv$. The use of the adjective \emph{strict} here is similar to~\citep{KK77}; the concept of banded stochastic matrices, however, seems to be new.

Notice also that Equation~\eqref{eq:limit} implies that  $\bF$ is regular (weakly-regular) if, and only if, $\bM^\infty\in\SSO_N$ ($\bM^\infty\in\SO_N$, respect.).

\subsection{Regularity and weak-regularity in the Kimura class}	

We begin by showing that Lemma~\ref{lem:stocordered} and the characterisation of weak-regularity in terms of $\bM^{\infty}$ already yields a sufficiency condition for a Kimura matrix to be weakly-regular:

\begin{proposition}
	\label{prop:wr}
	Let $\bM$ be a $(N+1)\times(N+1)$ Kimura matrix. If $\bM$ is eventually stochastically ordered then $\bM$ is weakly-regular.
\end{proposition}

\begin{proof}
 Let $\kappa_0\in\N$ be such that $\bN\bydef\bM^{\kappa_0}\in\SO_{N+1}$. Then Equation~\eqref{eq:limit} and Lemma~\ref{lem:stocordered} imply that  $\bM^\infty=\lim_{\kappa\to\infty}\bM^\kappa=\lim_{\kappa\to\infty}\bN^\kappa\in\SO_{N+1}$ , and hence that $\bF$ is non-decreasing.
\end{proof}

The following example shows that $\bM$ can be weakly-regular, without being eventually stochastically ordered. In particular, is not possible to characterise weakly-regular processes as those which become stochastically ordered in finite time.

\begin{example}
Let
\begin{equation*}
\bM=	\begin{pmatrix}
		1 & 0 & 0 & 0 \\
		\frac{1}{8} & \frac{1}{2} & \frac{1}{4} & \frac{1}{8} \\
		0 & \frac{1}{2} & \frac{1}{2} & 0 \\
		0 & 0 & 0 & 1 \\
	\end{pmatrix}.
\end{equation*}
It is easily verified that the corresponding fixation vector is $\bF=\begin{pmatrix}
0&\dfrac{1}{2}&\dfrac{1}{2}&1
\end{pmatrix}$, hence $\bM$ is weakly-regular. We check directly that 
\[
 \bM^\kappa=\begin{pmatrix}
             1 & 0 & 0 & 0\\
             \alpha_\kappa & \delta_\kappa & \gamma_\kappa & \alpha_\kappa\\
             \beta_\kappa & 2\gamma_\kappa & \delta_\kappa & \beta_\kappa\\
             0 & 0 & 0 & 1
            \end{pmatrix}
\]
for certain sequences $\alpha_\kappa,\beta_\kappa,\gamma_\kappa,\delta_\kappa$ defined recursively. In particular, 
\[
 \alpha_{\kappa+1}=\frac{1}{8}+\frac{2\alpha_\kappa+\beta_\kappa}{4}\ ,\quad \beta_{\kappa+1}=\frac{\alpha_\kappa+\beta_\kappa}{2}\ .
\]
It is easily verified by induction in $\kappa$ that $\alpha_\kappa,\beta_\kappa<\sfrac{1}{2}$. On the other hand,
\[
 \alpha_{\kappa+1}-\beta_{\kappa+1}=\frac{1-2\beta_\kappa}{8}>0\ ,
\]
and this proves that $\bM^\kappa\not\in\SO_{N+1}$ for $\kappa\in\N$.
\end{example}

On the other hand, it turns out that the regular processes are exactly those which become strictly stochastically ordered in finite time:

\begin{theorem}
	\label{thm:rc}
	Let $\bM$ be a $(N+1)\times(N+1)$ Kimura matrix. Then $\bM$ is regular if, and only if, it is eventually strictly stochastically ordered.
\end{theorem}

	We will prove here only necessity. The sufficiency proof is somewhat more involved, and we defer it to Section~\ref{ssec:regular_equiv} where a complete proof using a different argument is given. 
	
\begin{proof}[only-if part]
	In view of Remark~\ref{rmk:kkt} and  the characterisation of regularity as  $\bM^\infty\in \SSO_{N+1}$, we observe that $\bT^{-1}\bM^\infty\bT$ has all its  entries positive---except for the 0th column and row. Hence, the same is also true for $\bT^{-1}\bM^k\bT$, when $k$ is  sufficiently large. Thus, there exists $k_0$ such that $\bM^k\in\SSO_{N+1}$, if $k\geq k_0$.
\end{proof}

\subsection{Proof of Theorem~\ref{thm:rc}}

\label{ssec:regular_equiv}

We begin by examining in more detail the structure of $\bT^{-1}\bM\bT$, when $\bM\in\K$: If either $i=0$ or $j=0$, we have $\left(\mathbf{T}^{-1}\mathbf{MT}\right)_{ij}=\delta_{ij}$, and for $i,j>0$, we have that
\[       
\left(\mathbf{T}^{-1}\mathbf{MT}\right)_{ij}=\sum_{k,l}T_{ik}^{-1}M_{kl}T_{lj}=\sum_lM_{il}T_{lj}-\sum_lM_{i-1,l}T_{lj}=\sum_{l=j}^NM_{il}-\sum_{l=j}^NM_{i-1,l}.
\]
Therefore, we have that 
\[
\bT^{-1}\bM\bT=\begin{pmatrix}
1&\bzero\\\bzero^{\dagger}&\bL
\end{pmatrix},
\]
In particular, $\bM$ is essentially stochastically ordered if, and only if, $\bL^{k_0}$ is a positive matrix, for some $k_0\in\N$---such a matrix $\bL$ is said to be eventually positive. In what follows, we will say that $\bL$ is the associated matrix of $\bM$.

It turns out that we can characterise when a matrix is  regular using a class of matrices that satisfy the conclusions of the Perron-Frobenius theorem. Following \citet{Johnsoh_Tarazaga_2004}, we denote by $\mathsf{PF}_N$ the set of $N\times N$ matrices $\bA$ that have the Perron-Frobenius property, i.e., 
\begin{enumerate}
	\item $\rho(\bA)$ is a simple eigenvalue;
	\item there exist positive right and left eigenvectors of $\bA$ associated to $\rho(\bA)$.
\end{enumerate}

The next result makes this characterisation precise:	

\begin{theorem}
	\label{thm:reg_iff_pfn}
	Let $\bM$ be a Kimura transition matrix, and let
	\[
	\bT^{-1}\bM\bT=\begin{pmatrix}
	1&\bzero\\\bzero^{\dagger}&\bL
	\end{pmatrix}.
	\]
	Then $\bM$ is regular if, and only if, $\bL\in\mathsf{PF}_N$ 
\end{theorem}
\begin{proof}
	Let $\bF$ be  the fixation vector associated to $\bM$, and let
	\[
	\bT^{-1}\bF^{\dagger}= (0\; \widetilde{\bG})^{\dagger}=(0\; G_1\; G_2\;\cdots\; G_N)^{\dagger},\quad G_i=F_i-F_{i-1}.
	\]
	Then $\bM$ is regular if, and only if, $\widetilde{\bG}$ is positive. 
	
	Moreover, from the structure of $\bT^{-1}\bM\bT$, we immediately have that $\rho(\bL)=1$, and that  it is a simple eigenvalue of $\bL$. Indeed, $\be_0$ is a right and left eigenvector for $\bT^{-1}\bM\bT$, and $(0\; \bone^\dagger)$ is a left eigenvector of $\bT^{-1}\bM\bT$, and thus $\bone$ is a left eigenvector of $\bL$. 
	
	On the other hand, $\widetilde{\bG}$ is a real right eigenvector of $\bL$. Since we always have $G_1>0$, we have that  $\bL\in \mathsf{PF}_N$ if, and only if, $\widetilde{\bG}$ is positive.
\end{proof}

\begin{corollary}
If $\bM\in\BSO_{N+1}$, then $\bM$ is regular.
\end{corollary}

\begin{proof}
The corresponding $\bL$ is non-negative, and at least the diagonal, main superdiagonal, and lower superdiagonal are positive, and hence $\bL$ is
primitive; in particular, $\bL\in\mathsf{PF}_N$.
\end{proof}

As observed in \citet{Noutsos_2006,Tarazaga_Raydan_Hurman_2001}, there are matrices with negative entries in $\mathsf{PF}_N$. However, it can be shown that $\bA\in \mathsf{PF}_N$ if, and only if,  $\bA$ is eventually positive---cf. \citep{Johnsoh_Tarazaga_2004}. With this final observation, we are now ready to provide the proof to Theorem~\ref{thm:rc}:

\begin{proof}[Theorem~\ref{thm:rc}]
	In view of the remark above, Theorem~\ref{thm:reg_iff_pfn} already shows that $\bM$ is regular if, and only if, the associated matrix $\bL$ is eventually positive. As already pointed out above, this happens if, and only if, $\bM$ is eventually strictly stochastically ordered. 
\end{proof}

\begin{remark}\label{rmk:row_column_L}
	Notice that for $j=1,\ldots,N$, we always have 
	\[
	\sum_{i=1}^N\left(\bT^{-1}\bM\bT\right)_{ij}=\sum_{k=j}^N\sum_{i=1}^N\left(M_{ik}-M_{i-1,k}\right)=\sum_{k=j}^N\left(M_{Nk}-M_{0k}\right)=1\ .
	\]
	In particular, this means that $\begin{pmatrix}
	1 & \mathbf{1}
	\end{pmatrix}^\dagger$
	is a left eigenvector of $\bT^{-1}\bM\bT$ associated to the eigenvalue one, and hence that $\mathbf{1}$ is a left eigenvalue of $\bL$ associated to the same eigenvalue. 
	Notice also that the row-wise sum of $\bL$ yields information on the difference of the marginal increase in the  expected frequency of the first type after one step:
	\begin{align*}
	\sum_{j=1}^N\left(\bT^{-1}\bM\bT\right)_{ij}&=\sum_{j=1}^N\sum_{k,l=0}^NT_{ik}^{-1}M_{kl}T_{lj}=\sum_{k,l=0}^NlT^{-1}_{ik}M_{kl}\\
	&=\sum_{k=0}^N\mathbb{E}[X_{\ell+1}|X_{\ell}=k]T^{-1}_{ik}=\mathbb{E}[X_{\ell+1}|X_{\ell}=i]-\mathbb{E}[X_{\ell+1}|X_{\ell}=i-1], i>0\ .
	\end{align*}
	In particular, this implies that $\lim_{k\to\infty}\bL^k$ is a matrix that is constant by rows with each row $i$ being $F_i-F_{i-1}$.
\end{remark}

\section{A study of fixation and regularity in a selected list of processes}
\label{sec:fix_reg_classical}

\subsection{Moran process is regularly universal}\label{ssec:Moran_regular}

		The Moran process is a special case of a more general class known as birth-death (BD) processes. A general BD process is characterised by a transition matrix that satisfy $M_{ij}=0$, if $|i-j|>1$. Let $X_k$ denote the corresponding population process. Then it is easily checked that 
			\[
			\E[X_{k+1}|X_{k}=i]= i + M_{i,i+1}-M_{i,i-1}.
			\]

If $\bM$ is a Kimura tri-diagonal matrix, then  the fixation vector $\bF$ is given by \citep{Karlin_Taylor_first,GS:1997}:
		\begin{equation}
		\label{eq:tri_fix}
		F_i=c^{-1}\sum_{l=1}^i\prod_{k=1}^l\frac{M_{k-1,k}}{M_{k+1,k}},\quad c=\sum_{l=1}^N\prod_{k=1}^{l-1}\frac{M_{k-1,k}}{M_{k+1,k}}.
	\end{equation}
		In particular, every tri-diagonal process is regular. We now turn to a more detailed study of Moran processes:
		
		\begin{theorem}
			\label{thm:Moran_fix}
			Let $\bF$ be an admissible fixation vector. Then $\bF$ is the fixation vector of some Moran process if, and only if, $\bF$ is increasing. Moreover, in the latter case, the type fixation probabilities of the Moran process that realises such a vector are given by
				\begin{equation}\label{eq:pj_fixation}
				p_i=\frac{i(F_i-F_{i-1})}{i(F_i-F_{i-1})+(N-i)(F_{i+1}-F_i)}\in(0,1), \quad i=1,\dots, N-1.
				\end{equation}
				
		\end{theorem}
		\begin{proof}
			Assume $\bF$ is the fixation vector of a Moran process. Then it is increasing as a special case of Equation~\eqref{eq:tri_fix}.
			
			Conversely, assume that $\bF$ is increasing, and recall that $\bF$ is a fixation vector if, and only if, it satisfies the recursion
						\[
						F_{i-1}\frac{i}{N}(1-p_i)+F_i\left(\frac{i}{N}p_i+\frac{N-i}{N}(1-p_i)\right)+F_{i+1}\frac{N-i}{N}p_i=F_i.
						\]
						This can be rewritten as 
						\begin{equation}
						\label{eqn:rec_fgj}
						(N-i)p_iG_i-i(1-p_i)G_{i-1}=0,\quad i=1,\ldots,N-1 ,
						\end{equation}
						where $G_i$ is the marginal gain in fixation:
						\begin{equation}
						\label{eqn:gj_def}
						G_i\bydef F_{i+1}-F_i.
						\end{equation}
						  Equation~\eqref{eqn:rec_fgj} is usually solved for $G_i$, which leads to a special case of Equation~\eqref{eq:tri_fix}. However, it can be also uniquely solved for $p_i$  yielding equation~(\ref{eq:pj_fixation}), which satisfies  $0<p_i<1$.
		\end{proof}

\begin{remark}
Equation~\eqref{eq:pj_fixation}  indicates that the marginal increasing in the fixation probability---i.e. the increase in fixation probability, when the frequency of type \X increase by one---can be understood as a reproductive fitness, cf. Equation~\eqref{eq:p_and_fit_rep}. More precisely, we let
\begin{align*}
& \varphi^{(\A)}(i)\bydef F_i-F_{i-1}>0\ ,\\
& \varphi^ {(\B)}(i)\bydef (1-F_i)-(1-F_{i+1})=F_{i+1}-F_i>0 .
\end{align*}
\end{remark}

	Although a BD process do not need to be stochastically ordered to be regular, the class of stochastically ordered matrices will be of interest when discussing time-inhomogeneous processes, and thus we will digress about this point. First, we observe that a BD process is banded stochastically ordered if, and only if, we have
			\[
			M_{i,i+1}+M_{i+1,i}<1,\quad i=1,\ldots,N-2.
			\]
In particular, the Moran process is banded stochastically ordered if, and only if,
		\[
		\left(1-\frac{i}{N}\right)p_i+\frac{i+1}{N}(1-p_{i+1})<1,\quad i=1,\ldots,N-2.
		\]	
		This immediately yields the following result
		\begin{lemma}\label{lem:MoranBSO}
			Let $\bM$ be the Moran matrix associated to type selection probability $\bp$.
			\begin{enumerate}
				\item If $\bp$ is increasing, we have that $\bM$ is banded stochastically ordered;
				\item If $\bp$ is such that $\sfrac{1}{(N+1)}<p_i<1-\sfrac{1}{(N+1)}$, then $\bM$ is banded stochastically ordered. 
			\end{enumerate}
		\end{lemma}

\begin{remark}\label{rmk:Metzler}
	We can use Theorem~\ref{thm:reg_iff_pfn} to prove that BD processes are regular even if they are not stochastically ordered, and  without using the explicit expression for the fixation probability.  Indeed, for such processes we have that $\bL$ is a tridiagonal matrix with the non-zero entries given by
	\begin{align*}
	L_{j,j+1}&=M_{j,j+1}\\
	L_{j+1,j}&=M_{j+1,j}\\
	L_{j,j}&=1-M_{j,j+1}-M_{j+1,j}
	\end{align*}
	with $j=0,\ldots,N-2$ in the first two equations, and $j=0,\ldots,N-1$ in the last equation.
	
	Thus $\bL$ is an irreducible matrix with non-negative off-diagonal elements, and hence is an irreducible Metzler matrix. Such a matrix has the Perron-Frobenius property \citep{Arrow,BP79}, and hence BD processes are regular.
\end{remark}

\subsection{Regularity and irregularity in the Wright Fisher process}

The discussion of the Wright-Fisher process requires more work. It turns out that a very useful tool will be the Bernstein polynomial associated to the fixation vector $\bF$, namely:

\begin{equation}
\Upsilon_{\bF}(p)\bydef\sum_{i=0}^NF_i\binom{N}{i}p^i(1-p)^{N-i}.
\label{eqn:bpf}
\end{equation}
It is easy to check that $\gamma_\bF(0)=0$, and $\gamma_\bF(1)=1$. Furthermore, if $\bF$ is increasing, then  $\Upsilon_{\bF}(p)$ is an increasing function in $[0,1]$---cf. \citep{Phillips:2003} or \citep{Gzyl2003}.

We are now in a position to characterise the regular WF processes:

\begin{theorem}\label{thm:WFregular_sse}
Let $\bM$ be the transition matrix  of the Wright Fisher process associated to the type selection probability vector $\bp$. The three conditions below are equivalent.
\begin{enumerate}
 \item\label{WFregular_1} The process $\bM$ is regular.
 \item\label{WFregular_2} The matrix $\bM$ is strictly stochastically ordered.
 \item\label{WFregular_3} The vector $\bp$ is increasing.
\end{enumerate}
\end{theorem}

\begin{proof}
 \ref{WFregular_3}$\Rightarrow$\ref{WFregular_2}. 
  Define $h_n(p)=\sum_{i=n}^N\binom{N}{i}p^i(1-p)^{N-i}$. It is clear that $h_0(p)=1$ and $h_0'(p)=0$ for all $p$. Note that for $p\in(0,1)$ and $n\ge 1$
 \[
  h_n'(p)=\frac{1}{p(1-p)}\sum_{i=n}^N\binom{N}{i}p^i(1-p)^{N-i}(i-pN)\ .
 \]
We define
\[
\gamma_n\bydef\frac{\sum_{i=n}^N\binom{N}{i}\frac{i}{N}p^i(1-p)^{N-i}}{\sum_{i=n}^N\binom{N}{i}p^i(1-p)^{N-i}}\ ,
\]
and $\gamma_0=p$. Furthermore, $\gamma_n$ is the mean of $\sfrac{i}{N}$, from $i=n$ to $i=N$ with probability distribution given by $\binom{N}{i}p^i(1-p)^{N-i}$, and therefore $\gamma_N>\gamma_{N-1}>\dots>\gamma_1>\gamma_0=p$. In particular
$h_n'(p)>0$ for $n>0$. From the fact that $\sum_{j=n}^NM_{ij}=h_n(p_i)>h_n(p_{i-1})$, for all $n>0$ (with equality for $n=0$), we conclude that $\bM$ is strictly stochastically ordered.

  \ref{WFregular_2}$\Rightarrow$\ref{WFregular_1}. It follows immediately from Theorem~\ref{thm:rc}.

 \ref{WFregular_1}$\Rightarrow$\ref{WFregular_3}. Since $\bM$ is regular, we have that the fixation vector $\bF$ is increasing. In this case, as we have already pointed out, $\Upsilon_{\bF}$ is increasing, with $\Upsilon_{\bF}(0)=0$ and $\Upsilon_{\bF}(1)=1$. Furthermore, $\Upsilon_{\bF}(p_i)=F_i$. We conclude that $p_0=1-p_N=0$ and that the vector $\bp=\Upsilon_{\bF}^{-1}\left(\bF\right)$ is strictly increasing.
\end{proof}

\begin{remark}
	The equivalence between conditions 2 and 3 in Theorem~\ref{thm:WFregular_sse} can be seen as the strict dominance equivalence version  of the  more classic dominance case---cf. 
 \cite[Equation (1.1)]{Klenke:Mattner:2010}.
\end{remark}	
The result above shows that not every choice of $\bp$ yields a regular WF process. This naturally leads to  the question of what is the class of non-regular fixation probabilities that the WF process can realise. The next result shows, perhaps surprisingly, that any admissible fixation can be realised, although not necessarily uniquely.

\begin{theorem}[Universality of the Wright-Fisher process]\label{thm:universalityWF}
Let $\bF$ be an admissible fixation vector. Then there exists at least one WF matrix that has $\bF$ as a fixation vector.  In addition,
if $\bF$ is increasing, than such WF matrix is unique. 
\end{theorem}

\begin{proof}
	Notice that a given $\bF$ is the fixation vector of the Wright-Fisher process defined by $\bp$ if, and only if, we have%
	\[
	\Upsilon_{\bF}(p_i)=F_i,\quad i=0,\ldots,N.
	\]
	On the other hand,  we have that $\Upsilon_{\bF}(0)=0$, and $\Upsilon_{\bF}(1)=1$. In addition, we always have, for $p\in(0,1)$,
	\[
	0< \Upsilon_{\bF}(p)< \sum_{i=0}^N\binom{N}{i}p^i(1-p)^{N-i}=(p+(1-p))^N=1.
	\]
	  Since $\Upsilon_{\bF}$ is continuous,  the intermediate value theorem implies that $\Upsilon_{\bF}$ is onto $[0,1]$.  Thus,  given any admissible fixation vector $\bF$, there exists at least one type selection probability vector $\bp$ such that $\Upsilon_{\bF}(p_i)=F_i$. Furthermore, if $\bF$ is increasing, then $\Upsilon_{\bF}$ is also increasing and therefore $\bp$ is uniquely defined.
\end{proof}

\subsection{Evolutionary game theory and regularity of WF processes}

Most of the cases of non-constant fitnesses functions studied in the mathematical literature considers fitnesses obtained from evolutionary game theory, where pay-off are computed using two-player games. This corresponds to affine fitnesses functions, which are the simplest class of non constant functions. In this framework, \citet[Lemma~1]{Imhof:Nowak:2006} 
have shown that the corresponding WF matrix is totally-positive of order 2, and hence it is monotone---cf. \citet[Remark~1.1]{KK77}. Therefore, Proposition~\ref{prop:wr} then implies that these processes are weakly-regular. We now strength this result, and show that in its simplest and traditional setting, 2-player games, evolutionary game theory leads to increasing type selection probabilities, and hence to regular WF processes. At the end of this section, however, we provide an example which shows the existence of non-regular fixation patterns  in  WF processes at  the next level of generalization, i.e., in three-players game theory (quadratic fitnesses functions).

\begin{proposition}\label{prop:WFJogos_regular}
	If fitnesses functions are positive and affine, then the type selection probability vector $\bp$ is increasing.
\end{proposition}

\begin{proof}
	
	Let $\Psi^{(\A)}(i)=ai+b(N-i)+\alpha$ and $\Psi^{(\B)}(i)=ci+d(N-i)+\beta$, with $a,b,c,d,\alpha,\beta>0$. 
	We will show that the sequence
	\[
	p_i\bydef\frac{i\Psi^{(\A)}(i)}{i\Psi^{(\A)}(i)+(N-i)\Psi^{(\B)}(i)}=\frac{i(ai+b(N-i)+\alpha)}{i(ai+b(N-i)+\alpha)+(N-i)(ci+d(N-i)+\beta)}
	\]
	is increasing. Initially, let us show that
	\[
	f(x)\bydef \frac{x}{1-x}\frac{ax+b(1-x)+\alpha'}{cx+d(1-x)+\beta'}\ ,
	\]
	with $\alpha'=\alpha/N$ and $\beta'=\beta/N$,
	is increasing in the interval $[0,1]$. Let $x=\sfrac{i}{N}$ and write
	\begin{align*}
	p_i&=\left(1+\frac{(N-i)(ci+d(N-i)+\beta)}{i(ai+b(N-i)+\alpha)}\right)^{-1}=\left(1+\frac{1-x}{x}\frac{cx+d(1-x)+\beta'}{ax+b(1-x)+\alpha'}\right)^{-1}\\
	&=\left(1+f(x)^{-1}\right)^{-1}\ .
	\end{align*}
	
	Differentiating $f$ (and dropping primes), we find
	\[
	f'(x)=\frac{g(x)}{(1-x)^2(cx+d(1-x)+\beta)^2}\ ,
	\]
	where 
	$g(x)=(d+\beta)(b-a)(x-1)^2+(c-d)(a+\alpha)x^2+(d+\beta)(a+\alpha)$.
	$f$ is increasing if and only if $g$ is positive. Function $g(x)$ is quadratic with $g(0)=(b+\alpha)(d+\beta)>0$ and $g(1)=(a+\alpha)(c+\beta)>0$.
	Furthermore, we define
	\[
	x_0\bydef \left(1+\frac{(a+\alpha)(c-d)}{(d+\beta)(b-a)}\right)^{-1}\ , 
	\]
	such that $g'(x_0)=0$. 
	If $x_0\not\in[0,1]$, then $g$ is monotone in $[0,1]$ and from the fact that it is positive on the borders, it will be positive everywhere.
	Now, assume that $x_0\in[0,1]$.
	We have that
	\[
	g(x_0)=\frac{(a+\alpha)((d-c)\alpha+(a-b)\beta+ad-bc)}{a-b}x_0\ ,
	\]
	and $\frac{(a+\alpha)(c-d)}{(d+\beta)(b-a)}>0$. Therefore
	\begin{enumerate}
		\item If $c-d>0$ and $b-a>0$. Then, $bc>ad$ and consequently $g(x_0)>0$.
		\item If $c-d<0$ and $b-a<0$. Then, $bc<ad$ and consequently $g(x_0)>0$.
	\end{enumerate}
	We conclude that $g(x)>0$ for all $x\in[0,1]$. We have that $f$ is increasing in $[0,1]$ and we conclude $p_i=\left(1+f\left(\sfrac{i}{N}\right)^{-1}\right)^{-1}$ is increasing.	
\end{proof}

\begin{corollary}
 If the fitnesses functions are positive affine, then the Wright-Fisher matrix is regular.
\end{corollary}

The next example shows that, if we depart from the realm of 2-player games, then we can {have WF processes that are not even weakly-regular:

\begin{example}[{\protect A non-regular three-player game}]
Let $\varphi^{(\A)}(x)=15-24x+10x^2$ and $\varphi^{(\B)}(x)=1+14x^2$, which are strictly positive in the interval $[0,1]$, 
	then $p_i$ given from~\eqref{eq:p_and_fit_rep} is not increasing. These functions can be obtained from 3-player game theory, with $a_0 = 15$,  $a_1 = 3$,  $a_2 = 1$, $b_0 = 1$, $b_1 = 1$, $b_2 = 15$, where $a_k$ ($b_k$) is the pay-off of a type \A (\B,  respectively) player against $k$ other players. Note that the relative fitness $\Psi^{(\A)}/\Psi^{(\B)}=\varphi^{(\A)}/\varphi^{(\B)}$ is decreasing and is associated to coexistence games (i.e., $\Psi^{(\A)}/\Psi^{(\B)}>1$ for $x$  near zero, and $\Psi^{(\A)}/\Psi^{(\B)}<1$ for  $x$ near one).
\end{example}

\subsection{Alternative processes}\label{ssec:further}

	We finish this section with some comments about other models presented in the literature. We begin by discussing two models introduced as alternative dynamics that belong to the Kimura class and for which all the theory developed so far applies directly. In the sequel, we discuss two models  that are  likely to be unrealistic for most biological populations, and hence should be largely taken as pedagogical examples. The first one belongs to class $\K_1$ and, as noted in Remark~\ref{rmk:L1}, all general theorems apply. The second belongs to the class $\K_0$, and it can be seen as process in the boundary of the Kimura class.

\paragraph{Pairwise Comparison}	
		Another Birth-Death process that has also been  used as a model of evolutionary dynamics is the so-called pairwise-comparison process~(PC)~\citep{Traulsen2007522}, whose transition matrix is given by :
		\begin{equation*}
		M_{ij}=\left\{
		\begin{array}{ll}
		\frac{i(N-i)}{N^2}(1-q_i)\ ,\quad&i=j+1\ ,\\
		\frac{(N-i)^2+i^2+i(N-i)}{N^2}\ ,&i=j\ ,\\
		\frac{i(N-i)}{N^2}q_i\ ,&i=j-1\ ,\\
		0\ ,&|i-j|>1\ .
		\end{array}
		\right.
		\end{equation*}
		where $q_i$ ($1-q_i$) is the probability that \A replaces \B (\B replaces \A, respec.) in a pair contest.
		This process satisfies 
		\[
		\E[X^{PC}_{k+1}|X^{PC}_k=i]=i+\frac{i(N-i)}{N^2}(2q_i-1)
		\]
		Hence a PC process is neutral when $q_i=\sfrac{1}{2}$.
		We point out that this fact  was implicitly stated in~\citet{AltrockTraulsen,Hilbe}. As a matter of fact, when the intensity of selection  converges to zero, the replacement probability converges to $\sfrac{1}{2}$. Moreover, it is easily verified that every PC process is banded stochastically ordered, and therefore regular, inasmuch as we have that
		\[
		M_{i,i+1}+M_{i+1,i}=\frac{i(N-i)}{N^2}q_i+\frac{(i+1)(N-i-1)}{N^2}(1-q_{i-1})<\frac{1}{4}+\frac{1}{4}<1\ .
		\]
		Furthermore, given any increasing fixation vector $\bF$  we have it is realised by a PC process upon choosing
		\[
		q_i=\frac{G_{i-1}}{G_i+G_{i-1}}\in(0,1),
		\]
		with $G_i$ given by Equation~\eqref{eqn:gj_def}. Finally, we observe that while $\mathbf{q}$ is also a  type selection probability, it is based on a sample over  pairs instead of a sample over the whole population

\paragraph{Eldon-Wakeley}

			This model was introduced in~\citet{EldonWakeley}. It is an intermediate model between the Moran and the Wright-Fisher process, in which at each time step,  one individual is selected to reproduce, according to a TSP vector $\bp$, 
			and begets $U-1$ new individuals, $U\in\{1,\dots,N\}$. The parent	persists, while its offspring replace $U-1$ individuals who are selected with equal probability to die among the remaining individuals. The original work studied the neutral case, i.e., $p_i=\sfrac{i}{N}$, when it can be easily checked that $F_i=\sfrac{i}{N}$. Using the notation of the current work, the transition matrix is given by 
			\[
			 M_{ij}=p_iM^{(1)}_{ij}+(1-p_i)M^{(2)}_{ij}\ ,
			\]
			where
			\begin{align*}
			M_{ij}^{(1)}&=\binom{N-i}{j-i}\binom{i-1}{U-1-j+i}\binom{N-1}{U-1}^{-1}\ ,\\
			M_{ij}^{(2)}&=\binom{i}{i-j}\binom{N-1-i}{U-1+j-i}\binom{N-1}{U-1}^{-1}.
\end{align*}
			We use that $\binom{a}{b}=0$ whenever $b<0$ or $b>a$, and therefore $\bM^{(1)}$ ($\bM^{(2)}$) is lower (upper, respect.) triangular matrix. 
			After some simplifications, we find that for $N+1\ge n\ge 1$, $0\le i\le N$, 
			\begin{align*}
			 \sum_{j=n}^N\left(M^{(1)}_{i+1,j}-M^{(1)}_{ij}\right)&=\frac{(n-U)\binom{N-U}{N-n}\binom{U-1}{U-n+i}}{(N-i)\binom{N-1}{i-1}}
			 \left\{\begin{array}{l}
			        >0\ ,\quad U+i\ge n\ge\max\{i+1,U+1\}\\
			 =0\ ,\quad \text{otherwise}\ ,
			        \end{array}\right.\\
			\sum_{j=n}^N\left(M^{(2)}_{i+1,j}-M^{(2)}_{ij}\right)&=\frac{\binom{N-i-2}{N-U-n}\binom{i}{n-1}}{\binom{N-1}{N-U}}
			\left\{\begin{array}{l}
			        >0\ ,\quad \min\{i+1,N-U\}\ge n\ge i+2-U\ ,\\
			        =0\ ,\quad\text{otherwise}\ .
			       \end{array}\right.
			\end{align*}
			The associated $\bL$ matrix is an irreducible Metzler matrix and from the discussion in Remark~\ref{rmk:Metzler}, we conclude the Eldon-Wakeley process is regular for any choice of the TSP $\bp$. Note that  $\bM\not\in\BSO_{N+1}$, however the restriction to the matrix $\bar\bM=(M_{ij})_{i,j=U,\dots,N-U}$ is BSO, with band of size $U-1$.

As a last remark, we observe that in the original work \citep{EldonWakeley}, the parameter $U$ is a random variable with values in $\{2,\ldots,N\}$. Therefore, in our notation, we shall consider a matrix $\bM=\sum_{U=2}^{N}\wp(U)\bM_U$, where $\bM_U$ is the transition matrix given above for a certain fixed $U$, and $\wp$ is a probability mass function of $U$. Recalling  that the Kimura class in convex, we have that both the full original model~\citep{EldonWakeley} and the particular case studied in~\citet{der2011generalized} are accounted for, and the regularity of the convex combination follows from the fact that a convex combination of Metzler matrices is always of Metzler type. Notice also that, for $U=2$, this model is closely related to the Moran model, the only difference being that a newborn cannot  replace its parent.
		
\paragraph{$\Lambda_1$-model} This model was introduced in~\citet{der2011generalized} and while, in the authors words, is not realistic for most populations, it can be seen as a process that is ``antipodal" to the Wright-Fisher and that can be used to understand extremal behaviour in the evolutionary class. 
In this process, at each time step either nothing happens or one individual replaces the entire population. Here, we show how a small variation of our approach would apply in this case. Given a TSP vector $\bp$, we define the generalised $\Lambda_1$-process by the matrix

\[
M_{ij}=\left\{
\begin{array}{ll}
\frac{1-p_i}{N},&i=0,\ldots,N,\; j=0;\\
1-\frac{1}{N},&i=j,\;i=1,\ldots,N-1;\\
\frac{p_i}{N},&i=0,\ldots,N,\; j=N.
\end{array}
\right.
\]
A direct calculation shows that $F_i=p_i$, and hence that this is model is regular if, and only if, the corresponding TSP is increasing. Also, another direct calculation shows that $\bM\in\SSO$ if, and only if, $\bp$ is increasing. Hence, Theorem~\ref{thm:WFregular_sse} also holds for the $\Lambda_1$ process.

We stress that, since this model belongs to the $\K_1$ class, all the generic results holds for the $\Lambda_1$ process, including the discussion on time-inhomogeneous processes in Section~\ref{sec:inhomogeneous}. Nevertheless, the reasons as why  this model is considered by their authors as ``unrealistic'' are possible the same as it is not included in the Kimura class, i.e., that mixed states are not necessarily accessible, even considering long time intervals, from any mixed initial condition.

\paragraph{Lethal mutation}
   The process such that $p_i=0$ for $i\in\{1,\dots,N-1\}$~\citep{Schuster_MDL}, is  non-Kimura. Assume, however a family of increasing TSPs $\bp^{(\epsilon)}$, $\epsilon\ge 0$ such that $\bp^{(\epsilon)}\to\bp^{(0)}=(0,0,\dots,1)$. For each value of $\epsilon>0$ we define the WF transition matrix $\bM^{(\epsilon)}$ and it is clear that $M^{(\epsilon)}_{ij}\to\delta_{0j}$ for $i<N$ and $M_{iN}^{(\epsilon)}\to\delta_{iN}$. The continuity of $\bF$ with respect to $\bM$ --- that follows from Proposition~\ref{prop:fixationvector} --- implies that, as $\epsilon\to0$, $\bF^{(\epsilon)}\to(0,0,\dots,1)$, which is the fixation probability of the matrix $\bM^{(0)}$. On the other hand, if we assume the Moran process, then $M^{(\epsilon)}_{i,i-1}\to \sfrac{i}{N}$, $M^{(\epsilon)}_{ii}\to\sfrac{(N-i)}{N}$ and $M^{(\epsilon)}_{i,i+1}\to 0$ for $i<N$. The fixation vector is exactly the same as the one for the WF process and the limit of the fixation vector when $\epsilon\to0$ is also the fixation vector of the limit matrix. We conclude that the limit models of both M and WF models, in the case of lethal mutation, can be understood as limits of Kimura regular models, and therefore they belong to the boundary of the set of regular matrices. It is clear from the previous discussion that the fixation vector $\bF=(0,0,\dots,0,1)$ is to be expected, independently of the precise way this model is built. The limit model belongs to the class $\K_0$ discussed in Remark~\ref{rmk:L1}, and therefore the fixation vector exists but is not admissible. Finally, notice that if type \A represents a lethal mutation, then the assumption of constant population is artificial when applied to state $i=N$, and this explains the discontinuity of fixation in the limit model.

\subsection{Regular and smooth fixation  for large populations}

\label{sec:weak}

	Given an increasing fixation vector $\bF$, Theorems~\ref{thm:Moran_fix} and \ref{thm:universalityWF} show that there are unique vectors $\bp_M$ and $\bp_{WF}$ that realises this fixation vector for the Moran process and  for the Wright-Fisher process, respectively. We now want to study the behaviour of type selection vectors, if the following conditions are met:

	\begin{enumerate}
		\item $F_i=\phi(\sfrac{i}{N})$, with $\phi:[0,1]\to[0,1]$ being sufficiently smooth;
		\item $N$ is large, but still finite.
	\end{enumerate}
	It turns out that such assumptions imply that the corresponding TSPs are close to the neutral ones, and hence that they are equivalent to assume the weak-selection regime. We begin with Moran process:

\begin{proposition}
	Let $\phi:[0,1]\to[0,1]$ be a $C^2$ function, and assume that $F_i=\phi(\sfrac{i}{N})$. Assume also that $N$ is sufficiently large. Then 
	\[
	p_i^{(N)}=\frac{i}{N}-\delta_N\frac{i}{N}\left(1-\frac{i}{N}\right)\frac{\phi''(\sfrac{i}{N})}{\phi'(\sfrac{i}{N})}+\mathcal{O}(\sfrac{1}{N^2}),
	\]
	with $0<\delta_N<\sfrac{2}{N}$, is such that the fixation vector associated with a Moran process given by the type selection probability $\bp$ is $\bF$.
\end{proposition}

	\begin{proof}
		From Theorem~\ref{thm:Moran_fix} and on using the mean value theorem, we have
		
		\begin{align*}
		p_i^{(N)}&=\frac{(\sfrac{i}{N})\phi'(x_1)}{\left(1-\sfrac{i}{N}\right)\phi'(x_2)+(\sfrac{i}{N})\phi'(x_1)}&&x_1\in(\sfrac{i-1}{N},\sfrac{i}{N}),\; x_2\in (\sfrac{i}{N},\sfrac{i+1}{N})\\
		&=\frac{\sfrac{i}{N}}{\left(1-\sfrac{i}{N}\right)\frac{\phi'(x_2)-\phi'(x_1)}{\phi'(x_1)}+1}\\
		&=\frac{i}{N}-\delta_N\frac{i}{N}\left(1-\frac{i}{N}\right)\frac{\phi''(\sfrac{i}{N})}{\phi'(\sfrac{i}{N})}+\mathcal{O}(\sfrac{1}{N^2}),
		\end{align*}
	with $0<\delta_N<\sfrac{2}{N}$.
	\end{proof}

\begin{remark}
	Notice that if we use the fixation probability yielded by Replicator-Diffusion equation, then $-\sfrac{\phi''(x)}{\phi'(x)}$ is the gradient of selection, i.e., the difference between fitnesses of types \A and \B, in the weak selection regime~\citep{ChalubSouza09b,ChalubSouza_2016}.
\end{remark}	
	
In order to deal with this question for the   WF process, we need a  result from approximation theory:
\begin{lemma}[See~\citet{Estep:2002}, Section 3.6]
	\label{thm:bernstein_id} 
	Let $f:[0,1]\to\R$ be a Lipschitz continuous function, with Lipschitz constant $K$, and let
	\[
	\Upsilon^N_f(x)=\sum_{i=0}^Nf(\sfrac{i}{N})B_{i,N}(x),\quad B_{i,N}(x)=\binom{N}{i}x^i(1-x)^{N-i}.
	\]
	Then
	\[
	\|f-\Upsilon^N_f\|_\infty\leq \frac{9K}{4N^{1/2}}.
	\]

\end{lemma}

This result implies in the following important result

\begin{theorem}[Continuity of fixation and weak-selection]
	Assume that the fixation probability is described by an increasing continuously differentiable function $\phi:[0,1]\to\R$, such that there exists constants $K_-,K_+>0$, with $K_-\leq\phi'(x)\leq K_+$, $x\in[0,1]$. Consider the WF process with a population of size $N$, and let $\bp^{(N)}$ be a vector of type selection probabilities such that the corresponding fixation vector $\bF^{(N)}$ satisfies $F^{(N)}_i=\phi(\sfrac{i}{N})$.  Then	
	\[
	\left|\frac{i}{N}-p^{(N)}_i\right|\leq \frac{9K_+}{4K_-N^{1/2}},\quad i=1,\ldots,N-1.
	\]	
\end{theorem}

\begin{proof}
	On one hand, we have by the mean value theorem that
	\[
	|\phi(\sfrac{i}{N})-\phi(p_i^{(N)})|=|\phi'(\bar{z})|\left|\frac{i}{N}-p_i^{(N)}\right|\geq K_-\left|\frac{i}{N}-p_i^{(N)}\right|,
	\]
	for some $\bar{z}$ in the open interval delimited by $\sfrac{i}{N}$ and $p^{(N)}_i$.
	On the other hand, we also have by Theorem~\ref{thm:bernstein_id} that
	\[
	\left|\Upsilon_\phi^N(p_i^{(N)})-\phi(p_i^{(N)})\right|\leq \frac{9K_+}{4N^{1/2}}\ .
	\]
	Since $\Upsilon_\phi^N(p_i^{(N)})=\phi(\sfrac{i}{N})$, the result follows.
\end{proof}

\section{Kimura invariance and regularity in time-inhomogeneous processes}
\label{sec:inhomogeneous}

\subsection{A convex semigroup of evolutionary matrices}

Understanding evolution in fluctuating environments has always been an important issue, and the first studies seem to date back at least to the works of \citet{Kimura_1954} and \citet{Haldane_1963}. The earliest works dealing with this problem in finite populations seem to go back to the early 1970's \citep{gillespie1972,gillespie1973,Karlin_Lieberman,Karlin_Levikson}; see also the review in~\citet{Felsenstein:1976}, and the chapters in \citet{Karlin_Taylor_second} and \cite{gillespie1991causes}. Recently, it has been gaining importance again~\citep{Ashcroft:etal:2014,lorenzi2015dissecting}. Time dependent evolutionary processes lead naturally to the study of products of transition matrices, and how it might impact on the property that all states are accessible, in a finite number of steps,from any transient state, and on the regularity of such processes.

We are thus led to consider whether the class discussed in Section~\ref{ssec:natural_class} is closed under products. Since the product of irreducible or primitive matrices is not necessarily irreducible nor primitive, such a closure is not to be expected. Indeed, given two tridiagonal matrices such that their core have null diagonal, they are irreducible, but their product --- a pentadiagonal matrix with super- and sub-diagonal identically zero in the kernel --- is not.
 Tridiagonal matrices with null diagonal is a well-known device used for computation biologists to speed up  simulations for fixation; however, it will turn out from the discussion below that this device will usually not work for time-inhomogeneous processes. In order to deal with these difficulties, we begin by  restricting the Kimura class as follows:

\begin{definition}[The Gillespie class of matrices]\label{def:Gillespie}
We say that a matrix $\mathbf{A}$ is totally indecomposable if there are no permutation matrices $\mathbf{P}$ and $\mathbf{Q}$ such that $\mathbf{PAQ}=\left(\begin{smallmatrix}\mathbf{B}&\mathbf{0}\\ \mathbf{C}&\mathbf{D}\end{smallmatrix}\right)$, with $\mathbf{B}$, $\mathbf{D}$ non-trivial square matrices and $\mathbf{0}$ the null matrix.	We say that a Kimura transition matrix $\bM$ is a \emph{Gillespie} matrix if $\widetilde{\bM}$ is totally indecomposable. The Gillespie class will be denoted by $\G$.
\end{definition}

\begin{remark}
 We termed the matrices in definition~\ref{def:Gillespie} after John H. Gillespie who, as far as we know,  was the first to systematically study time-inhomogeneous evolutionary processes for finite populations using stochastic techniques, cf.~\citet{gillespie1972,gillespie1973}; see also~\citet{gillespie1991causes}.
\end{remark}

Totally indecomposable matrices can be also characterised as irreducible matrices that have positive diagonal \citep{lewin1971}. Such a characterisation immediately implies that the product and convex combinations of totally indecomposable matrix are again totally indecomposable, and this leads to the following result:
\begin{proposition}\label{prop:Kimura_convex}
	The class of Gillespie matrices is a convex set and it is closed by multiplication. In particular, it is a convex semigroup---cf. Remark~\ref{rmk:csg}.
\end{proposition}
\begin{remark}  \label{rmk:convex}
	\
	\begin{enumerate}
	\item The idea behind the use of totally indecomposable products is the fact that there is no two subset of non homogeneous states, i.e., $I_1,I_2\subset\{1,\dots,N-1\}$, $I_1,I_2\ne\emptyset$ such that $I_2$ cannot be reached from any state in $I_1$. 
		\item Any totally indecomposable matrix is primitive, and in particular irreducible \citep{lewin1971}. Thus,  the Gillespie class is contained in the Kimura class, and all previous results apply. On the other hand, all examples of Kimura matrices studied so far are also Gillespie matrices.
 
  \item We do not address the question if $\G$ is the largest convex semigroup contained in $\K$. Notice, however, that $\G$ is not the largest set where non-homogeneous processes are well behaved, since for any $\bM\in\K$ we have $\bM\cdot\G\subset\K$. In particular, any product of Kimura matrices, with all but one factors being Gillespie is Kimura, but not necessarily Gillespie.
\end{enumerate}
\end{remark}

Within this section, all matrices will be assumed to belong to the Gillespie class, unless stated otherwise.

\subsection{Periodic environments}
\label{ssec:periodic}

One feature of periodic varying environments is that, in general,  the fixation probability depends not only on the initial frequency, but also on the current time state of the environment. Indeed, consider a periodic environment of period $l$, and let the corresponding transition matrices be $\bM_0,\ldots,\bM_{l-1}$. We extend the indices of the matrices for all integers, such that if $k\equiv k'\,\mathrm{mod}\, l$, then $\bM_k=\bM_{k'}$. Let
    \[
     \bP_k=\bM_k\bM_{k+1}\cdots\bM_{k+l-2}\bM_{k+l-1}\ ,
    \]
i.e., $\bP_k$ identifies the products of transition matrices after $l$ steps (one period) when the process starts with $\bM_k$. 
We also define $\bP_\infty\bydef\lim_{n\to\infty}\bP_0^n$.

In the following result, we prove that the joint fixation probability will not depend on the initial condition, if and only if the fixation probabilities associated to all instantaneous evolutions are the same.

	\begin{lemma}
	{W}e have that
	\begin{equation}
	\label{eqn:prod_per}
	\lim_{n\to\infty}\bP_k^n=\bP_\infty,\quad k=1,\ldots,l-1,
	\end{equation}
	if, and only if, the fixation vectors of $\bM_k$, $k=0,\ldots,l-1$ are the same.
	\end{lemma}

	\begin{proof}
		First, we recall that if $\bA$ and $\bB$ are two Gillespie matrices, then we have that $\lim_{k\to\infty}\bA^k=\lim_{k\to\infty}\bB^k$ if, and only if, they have the same fixation probability. Note also that due to the semigroup property, the fact that all $\bM_k$ matrices are Gillespie implies that all $\bP_k$ are Gillespie.
		
		$\Leftarrow$ Let $\bF$ be such that $\bM_k\bF^{\dagger}=\bF^{\dagger}$ for all $k$, with $F_0=1-F_N=0$. It is immediate that $\bP_k\bF^{\dagger}=\bF^{\dagger}$ for all $k$, and hence  Equation~\eqref{eqn:prod_per} holds.

		$\Rightarrow$ First, we notice that $\bP_k\bM_{k+l}=\bM_k\bP_{k+1}$, and therefore $\bP_{k}^n\bM_{k}=\bP_k^n\bM_{k+l}=\bM_k\bP_{k+1}^n$. Taking $n\to\infty$ and using Equation~\eqref{eqn:prod_per} yields
	\begin{equation}
	\label{eqn:per_prod_com}
	\bP_\infty\bM_k=\bM_k\bP_\infty,\quad k=0,\ldots,l-1.
	\end{equation}
	Let
 	\[
	\bM_k=\begin{pmatrix}
	1&\bzero&0\\
	\ba_k^{\dagger}&\widetilde{\bM}_k&\bb_k^{\dagger}\\
	0&\bzero&1
	\end{pmatrix}
	\quad\text{and}\quad
	\bP_\infty=\begin{pmatrix}
	1&\bzero&0\\
	(\mathbf{1}-\widetilde{\bF})^{\dagger}&\bzero&\widetilde{\bF}^{\dagger}\\
	0&\bzero&1
	\end{pmatrix},
	\]
	with $\bF=(0,\widetilde{\bF},1)$ the fixation vector of all $\bP_k$.
	Together with the identity in Equation~\eqref{eqn:per_prod_com} we find
	\[
\widetilde{\bM}_k\widetilde{\bF}^{\dagger}+\bb_k^{\dagger}=\widetilde{\bF}^{\dagger}
	\] 
	for all $k$, and hence that $\bF$ is the fixation vector of $\bM_k$.		
	\end{proof}

	There is still one caveat before finishing this subsection.  The fact that $\G$ is a convex semigroup   implies  that for any  finite sequence $\bM_0,\ldots,\bM_{l-1}\in\G$, we have that  $\bP\bydef\bM_0\cdots\bM_{l-1} \in\G$, and that  $\bP^n\to \bP^\infty$, when $n\to\infty$. Nothing is said about the convergence of the product
	  \begin{equation}\label{eq:Mprod}
	  \prod_{k=0}^\infty\bM_k.
	  \end{equation}

	  In order to investigate such convergence, first we observe that all the partial products are of the form
	  \[
	  \dots,\,\bP_0^n,\,\bP_0^n\bM_{0},\,\bP_0^n\bM_{0}\bM_{1},\ldots,\bP_0^n\bM_{0}\cdots\bM_{l-2},\,\bP_0^{n+1},\ldots.
	  \]
	  A necessary and sufficient condition for convergence of the product is that all terms in the above equation converge to the same limit, i.e. the infinite product converges if, and only if, the following equalities are satisfied
	  \[
	  \bP_{\infty}=\bP_\infty\bM_0=\bP_\infty\bM_0\bM_1=\cdots=\bP_\infty\bM_0\cdots\bM_{l-2}.
	  \]
 	  But since $\bP_\infty\bM_k=\bP_\infty$, these equalities are all satisfied,  and  convergence follows.
	  
	 Notice, however, that this argument does not work for non-periodic products of the matrices $\bM_k$ --- in particular, it does not apply for random products. This will be discussed in Subsection~\ref{ssec:aperiodic}.

\subsection{Regular and non-regular evolution}

We now address the following question: assume that a time inhomogeneous process is stepwise regular. Is such a process itself regular?  

In order to answer this question, we begin with the following result

\begin{corollary}\label{cor:semigroup_regular}
The intersection of the set of banded stochastically ordered matrices with the set of regular Gillespie matrices is a convex semigroup.
Furthermore, let $\mathcal{R}$ be one of the following set of matrices:
 \begin{enumerate}
  \item WF matrices with increasing $\bp$ (or, equivalently, regular WF matrices).
  \item M matrices with increasing $\bp$.
  \item M matrices with $\bp\in (\epsilon_N,1-\epsilon_N),\quad \epsilon_N=\sfrac{1}{(N+1)}$.
  \item The union of any two of the previous sets or of all three.
 \end{enumerate}
Then the set generated by convex combinations and finite products of elements of $\mathcal{R}$ is a convex sub-semigroup of regular Gillespie matrices.
\end{corollary}

\begin{proof}
 The result follows from Lemma~\ref{lem:stocordered}, Proposition~\ref{prop:Kimura_convex} and Theorem~\ref{thm:rc}.
\end{proof}

\begin{example}
As an example of the previous corolary, we consider a Moran process in which \A dominates \B, with (TSP) given by $p_i=\left(\frac{i}{N}\right)^{1+\sfrac{1}{N}}$ and a Wright-Fisher process with a small dominance of \B over \A, with TSP given by $p'_i=\left(\frac{i}{N}\right)^{1-\sfrac{1}{N^2}}$. We immediately know that the joint process is regular. Indeed, for $N=10$, $\bF\approx
(0,                                    0.103,
                                    0.204,
                                    0.305,
                                    0.405,
                                    0.504,
                                    0.604,
                                    0.703,
                                    0.802,
                                    0.901,1)$, which shows a slight dominance from \A over \B. On the other hand, consider two regular WF processes, in the first the relative fitness is constant and equal to 1.7, while in the second case we have a frequency dependent relative fitness given by $i/N$. Note that, again we have that \A dominates \B in the first case and the reverse in the second process. Both TSP are increasing, and so it is in the case of the product process. However (again for $N=10$), the product process has a fixation probability typical from a coordination game (initially below neutral, eventually above). See Fig.~\ref{fig:fix_prod}.
\end{example}

\begin{figure}
\centering \includegraphics[width=0.45\textwidth]{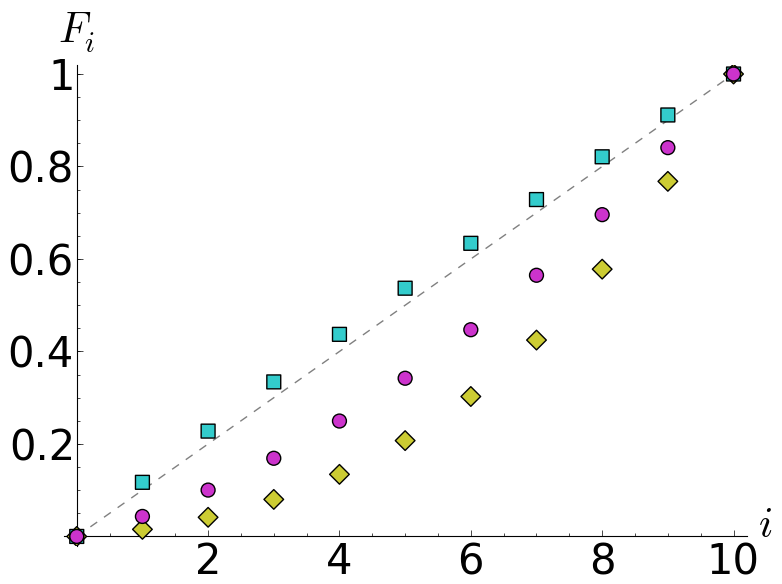}
 \includegraphics[width=0.45\textwidth]{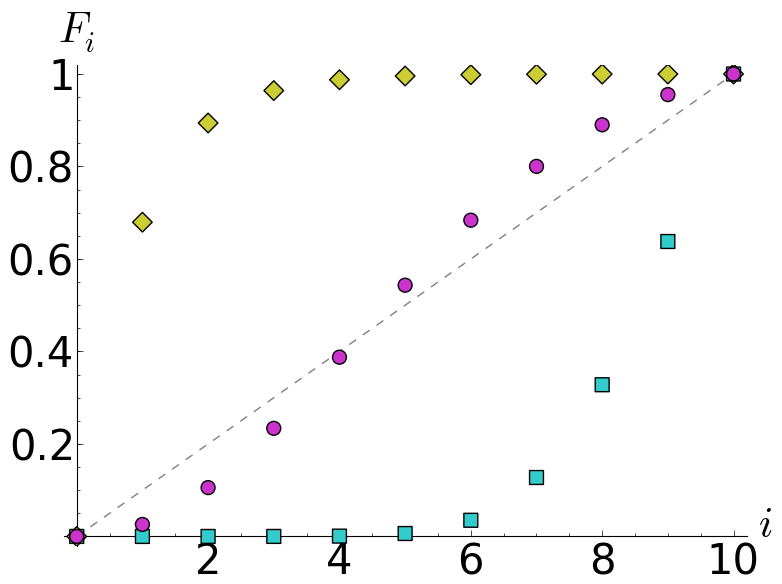}
 \caption{Fixation probability $F_i$ of the Wright-Fisher processes with respect to the initial condition $i=0,\dots,10$, for type selection probabilities $p_i=\left(\frac{i}{N}\right)^{1+\sfrac{1}{N}}$ and $p'_i=\left(\frac{i}{N}\right)^{1-\sfrac{1}{N^2}}$(left) and for relative fitnesses $\rho_i=1.7$, $\rho_i'=\sfrac{i}{N}$. In both cases the yellow diamond marks the fixation probability of matrix $\bM_1$, the cyan square for the matrix $\bM_2$ and the pink circle for the product matrix $\bM_1\bM_2$.}
 \label{fig:fix_prod}
\end{figure}

\begin{remark}
Notice that Theorem~\ref{thm:WFregular_sse} and Corollary~\ref{cor:semigroup_regular} show that a non-homogeneous WF  process, i.e., a process instantaneously given by a WF matrix, that is regular at  every instant   is itself regular. In particular, time inhomogeneity alone cannot induce non-regular fixation patterns --- cf.~\citet{gillespie1991causes}.
\end{remark}

For the Moran, and for tri-diagonal processes in general, multiplication outside of the class identified in Corollary~\ref{cor:semigroup_regular} can produce non-regular matrices, as shown in the next example.
\begin{example}
\label{examp:Parrondo}
Let
\[
\bp_1=\begin{pmatrix}0,\frac{1}{7},\frac{6}{7},1\end{pmatrix}
\quad\text{and}\quad
\bp_2=\begin{pmatrix}0,\frac{6}{7},\frac{1}{7},1\end{pmatrix}
\]
Then the corresponding Moran matrices are given by
\[
\bM_1=
\begin{pmatrix}
1 & 0 & 0 & 0 \\
\frac{2}{7} & \frac{13}{21} & \frac{2}{21} & 0 \\
0 & \frac{2}{21} & \frac{13}{21} & \frac{2}{7} \\
0 & 0 & 0 & 1 \\
\end{pmatrix}
\text{  and  }
\bM_2=
\begin{pmatrix}
1 & 0 & 0 & 0 \\
\frac{1}{21} & \frac{8}{21} & \frac{4}{7} & 0 \\
0 & \frac{4}{7} & \frac{8}{21} & \frac{1}{21} \\
0 & 0 & 0 & 1 \\
\end{pmatrix}.
\]
Let
\[
\bM_3=\bM_1\bM_2=
\begin{pmatrix}
1 & 0 & 0 & 0 \\
\frac{23}{147} & \frac{128}{441} & \frac{172}{441} & \frac{8}{49} \\
\frac{8}{49} & \frac{172}{441} & \frac{128}{441} & \frac{23}{147} \\
0 & 0 & 0 & 1 \\
\end{pmatrix},
\]
and let $\bF_i$, $i=1,2,3$,  be the corresponding fixation vectors. It is easy, though tedious, to check that
\begin{align*}
\bF_1&=\begin{pmatrix}0,\frac{1}{5},\frac{4}{5},1\end{pmatrix}^{\dagger}\\
\bF_2&=\begin{pmatrix}0,\frac{12}{25},\frac{13}{25},1\end{pmatrix}^{\dagger}\\
\bF_3&=\begin{pmatrix}0, \frac{244}{485}, \frac{241}{485}, 1\end{pmatrix}^{\dagger}.
\end{align*}
Hence $\bM_3$ is not regular despite the fact that $\bM_1$ and $\bM_2$ are regular. 
\end{example}

Example~\ref{examp:Parrondo} can be generalised straightforwardly  any even dimensions (and also adapted to odd ones). Indeed, assume $N+1$ is even let $\bp_1=(0,\sfrac{1}{k},\sfrac{(k-1)}{k},\ldots,\sfrac{1}{k},\sfrac{(k-1)}{k},1)$ and $\bp_2=(0,\sfrac{(k-1)}{k},\sfrac{1}{k},\ldots,\sfrac{(k-1)}{k},\sfrac{1}{k},1)$. Then for sufficient large $k$ the product of $\bM_1$ and $\bM_2$ is non-regular.

		The above example  yields two different deterministic versions in evolutionary dynamics of Parrondo's paradox in economy \citep{Parrondo_1,Parrondo_2}. In this sense, if evolution is described by $\bM_1$ in summer, and by $\bM_2$ in winter, then we obtain the following conclusions:
	\begin{enumerate}
		\item In an environment that is always summer or winter, the fixation probability of type \A with two individuals is larger than with just one individual. However, in the switching case the reverse holds. 
See~\citet{Osipovich} for a similar finding in biochemistry and~\citet{WilliamsHastings_AmNat2013} for a Parrondo paradox in ecology; in this second case there are two patches unable to sustain a certain populations; however if migration is allowed in alternatively between the two patches, the population may persist. See also~\citet{PeacockLopez20113124} for a direct example in which a Parrondo's game is directly related to seasonality. 

		\item The probability of fixation when there is only one individual in the switching case is larger than in every state of the non-switching case, while when there are two individuals it is smaller. This is the same ``surprising effect'' presented in \citet{Ashcroft:etal:2014}, however with deterministic time evolution, i.e, no stochasticity was assumed in the time evolution of the model. See also \citet{Melbinger:Vergassola:2015,Yakushina}.
	\end{enumerate}

\begin{remark}\label{rmk:barbosa}
Let us define a matrix $\mathbf{N}$ such that $N_{ii}=0$ and $N_{ij}=\frac{M_{ij}}{\sum_{k\ne i}M_{ik}}$ for $i\ne j$ and let $\bF$ be the fixation probability associated to $\bM$. It is immediate to prove that $\bF$ is such that $\mathbf{N}\bF=\bF$. In effect
\[
 \sum_jN_{ij}F_j=\frac{\sum_{j\ne i}M_{ij}F_j}{\sum_{k\ne i}M_{ik}}=\frac{F_i-M_{ii}F_i}{1-M_{ii}}=F_i\ .
\]
This observation has been used to speed up numerical computations of the vector $\bF$. See, e.g.~\citep{Barbosa_PhysRevE}. 

On the other hand, given two Moran matrices $\bM_1$ and $\bM_2$, with fixation vectors $\bF_1$ and $\bF_2$, one obtains matrices $\bN_1$ and $\bN_2$, using the procedure described above, that have the same fixation vectors. Nevertheless, the fixation vectors of $\bM_1\bM_2$ and $\bN_1\bN_2$ will be different, unless in very special cases. 
 \end{remark}

\subsection{Aperiodic evolution}\label{ssec:aperiodic}

As observed in Subsection~\ref{ssec:periodic}, the convergence of non-periodic products is not guaranteed by the semigroup property, despite the fact that all finite products belong to $\G$.  However, we will now show that the  results in \citet{DL:1992} (see also \citep{Bru:etal:1994}) can be applied in our present context. Following \citet{DL:1992}, we say that a set $\mathcal{S}=\{\bM_0,\ldots,\bM_{l-1}\}$ is a  Right Convergent Product (\RCP) set if for any sequence of integers  $\{d_i\}_{i=1}^\infty$ with $0\leq d_i\leq l-1$ we have that the right product
\[
\prod_{i=1}^\infty\bM_{d_i}=\bM_{d_1}\bM_{d_2}\cdots\bM_{d_n}\cdots
\]

is well defined.  Let us also write
\[
\bT^{-1}\bM_k\bT=\begin{pmatrix}
1&\bzero\\
\bzero^\dagger&\bL_k
\end{pmatrix},
\]
and therefore the set $\mathcal{S}$ may be defined by the $(N+1)\times(N+1)$ matrices $\bM_k$ or by the $N\times N$ matrices $\bL_k$.

We now show the following result: 
\begin{lemma}
	\label{lem:aper_prod}
Let $\mathcal{S}$ be a finite set with  $\mathcal{S}\subset \SO_{N+1}$. Assume  that there exists an integer $m$, such that for all $k\geq m$ we have that all products $\bL_{d_1}\bL_{d_2}\cdots\bL_{d_k}$ have a positive row. Then $\mathcal{S}$ is an \RCP\ set.  
\end{lemma}

\begin{proof}
	Under the assumptions, the matrices $\bL_k$ are column stochastic matrices with a single eigenvalue $\lambda=1$, cf. Remark~\ref{rmk:row_column_L}. The existence of the integer $m$ implies  in condition (C4)  of Theorem 6.1 of \citep{DL:1992}, from which the result follows.
\end{proof}
As a special case of the Lemma~\ref{lem:aper_prod}, we observe that if $\mathcal{S}\subset \BSO$ then every product $\prod_{i=1}^n\bM_{d_i}$ is positive, for $n>N$, and hence the set is \RCP. Thus, the same condition that guarantees regularity also ensures that arbitrary products drawn from a finite set exists.

From a more general perspective, if the set $\mathcal{S}$ is infinite, and if the factors are drawn following an stationary ergodic stochastic process then, since  all matrices are stochastic, their product converges almost surely---cf. \citep{Hennion:1997}. Naturally, if $\mathcal{S}\subset\G$, then the limit will be in $\G$.  Notice that the conditions required for \RCP\ here are more restrictive, however  the conclusions are stronger in the sense that every infinite product is definite and not just almost all. Notice also, that under the hypothesis we have that the  limit function, as defined in \citet{DL:1992}, is continuous.

\begin{remark}[Limit function]
	As observed above, under the assumptions of Lemma~\ref{lem:aper_prod}, if  $|\mathcal{S}|=l$, and if we write $\mathcal{S}_l$ to denote the set of sequences $\fd=\{d_k\}_{k=1}^\infty$, with $d_k\in\{0,1,\ldots,l-1\}$, endowed with the metric $D(\fd,\fd)\bydef l^{-r}$, where $r$ is the first index such that $d_r\ne d_r'$. Then there exists a continuous matrix function $\bM_\infty:\mathcal{S}_l\to M_{N+1}(\R)$ describing all the possible products with elements from $\mathcal{S}$. Namely, we have that
	\[
	\prod_{k=1}^\infty \bM_{d_k}=\bM_\infty(\fd)
	=
\begin{pmatrix}
	1&\bzero& 0\\
	1-F_1(\fd)&\quad\bzero\quad&F_1(\fd)\\
	1-F_2(\fd)&\bzero&F_2(\fd)\\
	\vdots&\vdots&\vdots\\
	0&\bzero&1
	\end{pmatrix}			
	\]
	In the present context, this implies that the functions $F_i:\mathcal{S}_l\to[0,1]$ are uniformly continuous in the topology induced in the space of sequences by the metric $D$, i.e., given $\epsilon>0$ there exists $\delta>0$ such that
	$|F_i(\fd)-F_i(\fd')|<\epsilon$, whenever $D(\fd,\fd')<\delta$.  This representation implies two results: (i) that the fixation probability converge uniformly along the product; (ii) after a sufficient large, but finite, number of steps, the fixation probability is known with large precision. In particular, one can compute the fixation probability of an arbitrary infinite product (random or not) of matrices drawn from an \RCP\ set using a finite and potentially  small sub-product.
	\end{remark}

\begin{remark}[Mixtures vs random products]\label{rmk:mixtures}
	Given a finite set $\mathcal{S}=\{\bM_1,\ldots,\bM_l\}\subset\G$, the convex combination of elements of $\mathcal{S}$  can be interpreted as a mixture of  matrices in $\mathcal{S}$. While mixtures enjoy a long tradition in statistics and other areas, their applicability in evolutionary dynamics seems to be first pointed out by \citet{Der:etal:2012}. Notice, however, that the stochastic processes associated to the transition matrices are somewhat more limited under the mixture approach than under the random matrix product approach.  As an example, take $l=2$, and let $X_k$, $k\in\N$ be i.i.d, with values $\bM_1$ or $\bM_2$, and binomial distribution of parameter $p$.Then 
	\[
	\E[X_k]=\E[X_1]=p\bM_1+(1-p)\bM_2.
	\] 
	Also, let $X^n=\prod_{k=1}^nX_k$. Then, because of independence, we have
	\[
	\E[X^n]=\E[X_1]^n
	\]
	Thus a mixture replaces the stochastic process for the transition matrices by its expected value under the assumption of independence, while a random product allows more generic behaviour.
\end{remark}

\section{Discussion}
\label{sec:discussion}

The contribution of this work goes along three main lines: (i) an axiomatisation of the algebraic properties of evolutionary processes; (ii) a qualitative study of fixation in finite populations, including the identification and characterisation of regularity; (iii) the study of how to compose basic processes in order to model environmental modifications, i.e., the  rigorous construction of time-inhomogeneous evolutionary processes. In this final section,  an unified view  of these three components is presented, showing how they fit into the general framework.

Indeed, revisiting the typical basic setup allowed us to introduce the parametrisation through the so-called type selection probabilities. They have the advantage of being directly accessible, and hence are particularly helpful for understanding basic principles of the models --- as some of the results obtained here already indicate. 

Along the second line, we  linked regular Wright-Fisher (WF) process and the newly introduced vector of type selection probabilities. Furthermore, we also showed the existence of non-regular WF processes.  WF process are associated to \emph{micro-evolution}, i.e., step-by-step evolution, while large jumps are possible, but rare~\citep{Charlesworth}. However, for certain choices of the frequency-dependent fitnesses functions, if a population has evolved to a monomorphic configuration (i.e., type \A, say, has fixed) it will be more likely that such a fixation occurred through a large jump from an intermediate step, than after a long and continuous process, where $x$, the fraction of type \A individuals in the population, on average, steadily increases towards 1 from smaller values. As an example, let us consider the case where environmental conditions have not changed (i.e., fitnesses functions are the same) since $t=0$, when the population was mixed, and at a later time we find a population in the  state $x=1$. What was the most probable state of the population at time $t=0$? On the absence of further intermediate measurements, and on using a maximum likelihood estimator, then the most probable state is given by $i^*=\mathop{\mathrm{arg\,max}}_{i\in\{1,\dots,N-1\}}F_i$. Assuming that all mixed states are equally probable at $t=0$, an application of Bayes' Theorem will give the same answer. Indeed, if $F_i$ is initially increasing, then it drops close to zero for larger values of $i$, and  it eventually increases  until $F_N=1$ only near $x=1$, then it is clear that fixation of type \A is possible, but it is more likely to happen if it avoids larger values of $i$ --- or, in simpler words, if it \emph{jumps} from intermediate values of $i$, when both types are present in comparable amounts, straight into fixation, i.e., with $i=N$.

According to Theorem~\ref{thm:WFregular_sse}, large jumps in the WF process will be possible only if fitnesses functions are not affine in $x$. Therefore, within the WF framework, truly multiplayer games might have much more complex dynamics: we will show elsewhere that any fixation pattern, and therefore any relative fitness, can be well approximated by pay-offs from $d$-player game theory, provided $d$ is large enough. Affine functions correspond to 2-player games; therefore, discontinuities in evolution (jumps) are associated to interactions in the population involving necessarily more than 2 individuals at the same time and that cannot be reduced to a series of pairwise interactions. One possible example is the evolution of the language~\citep{Atkinson01022008}. As human evolution is regulated by complex social interaction~\citep{Mathew:2015}, we may expect frequent discontinuities in evolutionary traits, specially if more types (i.e., pure strategies in a game) are allowed~\citep{GokhaleTraulsen2010}.

The topic of small \emph{versus} large changes in evolution (or, in other words, the compatibility between micro- and macro-evolution) is an import one. This is the traditional dichotomy between the gradualist view and the punctuational  view of evolution. See, for example, the discussion in~\citet{Frazzetta,Charlesworth_Lande_Slatkin} and references therein. 
As explained in \citet{Frazzetta}  ``large steps in evolution are more infrequent than small ones (...). But that fact alone cannot be used to dismiss large-step change.'' More precisely 22\% of  substitutional changes at the DNA level can be attributed to punctuational evolution~\citep{Pagel_Venditti_2006}. In view of the discussion in section~\ref{sec:regular}, there is no incompatibility between models used primarily for the study microevolution (the \emph{Fisher's microscope}~\citep{Waxman_Welch,Frazzetta}) and jumps in the evolutionary process. Here we discuss discontinuous evolution from mixed populations to homogeneous one, without intermediate mixed states, but the word \emph{macroevolution} has many different meanings~\citep{Erwin}; our approach  describe discontinuities in the fossil record~\citep{Frazzetta}, not speciation~\citep{Erwin}.

As a consequence of the discussion in Subsection~\ref{sec:weak}, if the population is large and the fixation probability is the restriction of an increasing smooth function (the same for all sufficiently large $N$), then we are are forced to be in the quasi-neutral (or weak-selection) regime. This conforms to the idea that an allele conferring  great advantage will typically have a large effect (see~\citep{Frazzetta} and references therein), alternatively, if the force of selection is small, the process is regular and the fixation probability is continuous: no jumps are allowed.

We also offer the construction of an algebraic framework to study theoretical population genetics: This idea is not new and can be traced back at least to the Ph.D thesis of Cotterman and Shannon~\citep{Cotterman:1940,Shannon} --- see \citet{Crow:2001} for an interesting historical perspective on these thesis. In particular, we formulate a general theory for evolutionary process in finite populations of haploid type, constant size, without mutations but considering very general effects from natural selection. Most of the modelling in population genetics assume a constant transition matrix between all possible states in a population; namely, they assume a choice of a certain stochastic process (like Moran, Wright-Fisher, pairwise comparison among many others). We want to be able to combine different processes. More precisely, given processes $\mathbf{M}_1$ and $\mathbf{M}_2$, we consider two different possibilities of combining them:
\begin{enumerate}
 \item The convex combination with parameter $\mu\in(0,1)$ of their transition matrices. This is what is known as a mixture of processes, and as observed in Remark~\ref{rmk:mixtures} replaces the random product of these matrices by the corresponding expected value. This yields a time homogeneous mean-field approximation of this stochastic behaviour.
 \item The product of their transition matrices. This represents a time inhomogeneous evolution, with the inhomogeneity being either deterministic --- as in the case of periodical variations --- or random.
\end{enumerate}

Mixtures, or convex combinations, have been already used in modelling evolution in~\citet{EldonWakeley}, but its importance in evolutionary processes seems to be first pointed out by~\citet{der2011generalized}. Inhomogeneous Markov chains have been considered in evolutionary models previously---a recent example is~\citet{Ashcroft:etal:2014}. On the other hand, we are not aware of a unified treatment and the identification of the underlying algebraic structure of transition matrices usually employed in modelling evolution dynamics --- convex semigroup of evolutionary matrices.

We paid particular attention to the behaviour of the fixation probability, in particular to the study of its monotonicity with respect to the initial frequency of a given type. We have also paid special care to guarantee that the composition of admissible processes is also admissible. This led us to the introduction of the Gillespie class as a set of matrices that is closed under multiplication (representing time inhomogeneity in evolution) and convex combinations (representing mixtures) that includes as particular cases the Moran process and the Wright-Fisher process. Furthermore, we define a subset where regularity is preserved under the same operations. Some of the qualitative results obtained in this class of matrices will not depend on the details of the modelling assumption, which is important because neither model is a first-principle model. We have also identified subclasses within the Gillespie class, which preserves regularity under composition.
 
 We also built an evolutionary dynamics version of the Parrondo's paradox (i.e., type \A has a larger fixation probability than type \B, given a certain initial condition, in two different environments but a smaller one in the switching environment) for the Moran process. Still for the Moran process, we showed a situation in which a larger initial frequency implies a large fixation probability in two static environment, but not in the switching case. This is not possible in the Wright-Fisher case. Parrondo's paradox also appear in models related to population genetics, although not directly based in any real example (as in the case of the present work). Considering a model for sexually antagonistic selection,~\citet{Reed_2007}  built an example of a two locus system with epistasis in which  an auotsomal allele can reach fixation  despite a lower average fitness of the alternative allele. In a more general setting Parrondo's paradox may also be used for the study of phenotypical switching~\citep{Fudenberg2012}.
 
This manuscript should also shed some light on the role of the neutral evolution. More specifically: assume a non-neutral Gillespie matrix $\bM$ and a neutral one $\bN$. The associated fixation vectors are $\bF_\bM$ and $\bF_\bN$, respectively, where $\bF_\bN=(0,N^{-1},2N^{-1},\dots,1)$. Stochastic processes given by $\alpha\bM+(1-\alpha)\bN$, $\bM\bN$ and $\bN\bM$ have, in general, fixation vectors that are different from $\bF_\bM$ (and, clearly, from $\bF_\bN$). In this sense, it shows that \emph{neutrality} is a property that strongly depends  on the environment and the interactions within the population. For instance, let us say that the evolution is given by the deterministic environment $\bN\bM$, and therefore we would like to say that in the odd steps, evolution is instantaneously neutral. However, as the effect of the neutral evolution is context-dependent, we may not extend it to the entire process and say that it is neutral part of the time, or that it is neutral with a certain probability. This makes the definition of neutrality used here closer to the concept of \emph{iso-neutrality}, and not to the stronger concept of euneutrality, cf.~\citet{Proulx}. 
 In~\citet{der2011generalized}, the conditional expectation defines neutrality, but the definition of a stronger concept, called ``pure drift (...) process'' requires that the variance is also comparable with of the neutral Wright-Fisher process (the same as in the case of the Moran process, up to multiplicative constants). Therefore, neutral matrices can induce distinctive behaviours in a a stochastic process depending, possibly, in the associated higher-order moments.
 
Diffusion approximations  have a long tradition in population genetics---cf. \citet{ChalubSouza14a} and references therein. More recently, a general limiting Kimura equation has been obtained~\citep{ChalubSouza09b,ChalubSouza14a} where it was termed the ``replicator-diffusion equation''. Considering what was discussed above, an important follow up of the current work would be the derivation of a diffusion approximation for time dependent fitness. Some derivations of time dependent Kimura equation appear in the literature, but they are obtained from semi-heuristic considerations, and not as large population limits of basic stochastic processes~\citep{Uecker,Cvijovic}. A particular question to be tackled in future works is how reversed dominance (i.e., type \A dominates \B in summer and the reverse holds in winter) might be able to generate metastable (quasi-stationary) intermediate states. This would suggest that the existence of metastable states might be a natural consequence of a changing environment, and is not necessarily (as sometimes claimed) a strategy that species develop to deal with uncertain future environment~\citep{Feldman}. This will clearly depend on the ratio between characteristic intergeneration and oscillatory time-scales, and the strength of stochastic effects determining the environment conditions.

\begin{acknowledgements}

FACCC was partially supported by FCT/Portugal Strategic Project UID/MAT/00297/2013 (Centro de Matemática e Aplicações, Universidade Nova de Lisboa) and by a ``Investigador FCT'' grant. FACCC is also indebted to Alexandre Baraviera (Universidade Federal do Rio Grande do Sul, Brazil) and Charles Johnson (College of William and Mary, USA) for useful discussions in preliminary ideas of this work. MOS was partially supported by CNPq under grants  \#~308113/2012-8, \# 486395/2013-8 and  \# 309079/2015-2. MOS also thanks the hospitality of the Universidade Nova de Lisboa and the partial support under grant UID/MAT/00297/2013. MOS further thanks preliminary discussions of some the ideas in this work  with the working group in evolutionary game theory at Universidade Federal Fluminense. Both authors also thank useful comments from Henry Laurie (Cape Town University), Alan Hastings (University of California at Davis), the handling editor, and an anonymous referee which helped to improve the original manuscript.

\end{acknowledgements}

\bibliographystyle{apa}

\end{document}